\documentclass[11pt]{article}
\usepackage{default}

\usepackage[margin=1in]{geometry}
\usepackage{amsmath,amssymb}
\usepackage{array}
\usepackage{hyperref}

\mathtoolsset{showonlyrefs}

\title{Resolving the Edge of a Quantum Pyramid}
\author{Alvan Arulandu\\Harvard University\\\href{mailto:aarulandu@college.harvard.edu}{\texttt{aarulandu@college.harvard.edu}}}

\begin{document}
\maketitle
\begin{abstract}
   Standing on the shoulders of giants, we resolve the quantum pyramids conjecture, confirming the globally information-optimal measurement for an ensemble of equiangular equiprobable pure states \citep{englert2010well}. We do so by proving the remaining entropy inequalities from \citep{holevo2025quantum} which certify optimality for obtuse and flat pyramids.  For obtuse pyramids, our key contribution is a rigorous proof that local minimizers of the corresponding entropy inequality can not have three distinct coordinate values. We show that eliminating this family can be reduced to a neat algebraic reciprocal inequality relating branches of the Lambert $W$ function, which may be of independent interest.
   For flat pyramids, we specifically prove a tight $\ell^p$ inequality for zero-sum vectors which was recently conjectured, analytically proved for dimension $d=3$, and computationally verified for $d\leq 200$ by \citep{holevo2026conjecturetightnorminequality}. We prove this bound for all $d\geq 2$ via a technique in symmetric inequalities known as the equal variables method.
\end{abstract}

\thispagestyle{empty}


\tableofcontents
\setcounter{page}{0}
\thispagestyle{empty}

\newpage

\section{Introduction}
The story of the quantum pyramid begins with a fundamental question in quantum information:
\begin{center}
    \textit{What is the maximal information that a measurement can extract \\from an ensemble of quantum states? }
\end{center}
This basic quantity, called the \textit{accessible information}, characterizes the ability of two parties to communicate over a perfect quantum channel using a particular ensemble of quantum states. Namely, suppose Alice chooses a message $j\in[m]$, encodes it as the quantum state $\rho_j $, and sends the state to Bob, who applies some measurement $\c M=\{M_k\}_{k\in [n]}$ to arrive at a decoded message $k\in [n]$. If the probability distribution $\{\pi_j\}_{j\in [m]}$ over messages as well as the collection of states $\{\rho_j\}_{j\in [m]}$ is known by both Alice and Bob, the accessible information of the \textit{quantum ensemble} $\c E =\{\pi_j,\rho_j\}_{j\in [m]}$ is simply the maximum Shannon mutual information between $j$ and $k$:
\begin{equation}\label{eq:A-sup}
    A({\c E})=\sup_{\c M }I({\c E},{\c M})
\end{equation}
Unfortunately, the accessible information is notoriously hard to calculate for an arbitrary ensemble, and progress has historically been limited to special symmetric ensembles. Even in these cases, the optimality of the hypothesized measurements is difficult to prove, with positive solutions only for commuting ensembles \citep{suzuki2008accessible}, binary ensembles \citep{levitin1995optimal, keil2024proof}, and highly symmetric ensembles such as the \textit{trine} and its siblings \citep{sasaki1999accessible}. 

However, breakthrough recent work considering suitable duals of \eqref{eq:A-sup} showed that the optimality criterion can be sufficiently expressed by certain entropy inequalities which tightly lower bound the classical Shannon information \citep{holevo2022optimization, holevo2025quantum}. For specific ensembles, this leads to non-trivial conjectured entropy inequalities that are ``distant relatives" of the famous log-Sobolev inequality which certify optimality. Referring the reader to \citep{holevo2025quantum} for a fuller discussion of this approach and its motivation, in this work, we prove the entropy inequalities arising from the accessible information problem for a particular class of symmetric pure states proposed by \citep{englert2010well}:
\begin{equation}\label{eq:pyramid-r01}
    \ket{\psi_j}\equiv \sqrt{mr_1}\ket{e_j}+(\sqrt{r_0}-\sqrt{r_1})\ket{e_0}\quad (j\in[m],r_0,r_1\geq 0,r_0+(m-1)r_1=1)
\end{equation}
where $\ket{e_0}\equiv \frac{1}{\sqrt m}\sum_{j=1}^m \ket{e_j}$. Equivalently,
\begin{equation}\label{eq:pyramid-ab}
    \ket{\psi_j}\equiv a\ket{e_j}+b\sum_{i\neq j}^m\ket{e_i}\quad (j\in[m], a\geq m^{-1/2}, a^2+(m-1)b^2=1)
\end{equation}
such that $r_0=(a+(m-1)b)^2/m$ and $r_1=(a-b)^2/m$ connect the two parameterizations. 
These equiangular equiprobable ensembles, called \textit{quantum pyramids}, are an important family related to orthogonal signals in communication theory \citep{helstrom1976quantum} and in fact contain ensembles such as the \textit{lifted trine} which was used by \citep{shor2002number} to argue that the number of POVM elements in the optimal measurement may exceed the number of states in the ensemble. Identifying each state as an edge of the quantum pyramid, we can classify the quantum pyramid ${\c E}=\left\{\frac{1}{m},\ketbra{\psi_j}\right\}_{j\in [m]}$ via the cosine of the angle between any two edges: $\xi\equiv \braket{\psi_j}{\psi_k}$ for $j\neq k$.
\begin{table}[H]
\centering
{
\renewcommand{\arraystretch}{1.75}
\begin{tabular}{c|c|c|c|c|c}
\textbf{Type} & \({\xi}\) & \({b}\) & \({a}\) & \(A({\c E})\)& \({\c M}^*\)\\ 
\hline
\text{ort} 
    & \(1\) 
    & \(\frac{1}{\sqrt m}\) 
    & \(\frac{1}{\sqrt m}\)
    & 0 & 
    \\
\text{orthogonal} 
    & \(0\) 
    & \(0\) 
    & $1$
    & \(\log_2 m\) & $\{\ketbra{e_k}\}_{k\in [m]}$\\
    
    \hline
\text{acute} 
    & \((0,1)\) 
    & \(\left(0,\frac{1}{\sqrt m}\right)\) 
    & \(\left(\frac 1 {\sqrt m},1\right)\)
    & \multicolumn{2}{c}{Corollary 1 of \citep{holevo2025quantum} }\\
\text{obtuse} 
    & \(\left(-\frac{1}{m-1},0\right)\) 
    & \(\left(-\frac{1}{\sqrt{m(m-1)}},0\right)\) 
    & \(\left(\sqrt{\frac{m-1}{m}},1\right)\)
    & \multicolumn{2}{c}{Theorem~\ref{thm:obtuse-59-67}}\\
\text{flat} 
    & \(-\frac{1}{m-1}\) 
    & \(-\frac{1}{\sqrt{m(m-1)}}\)
    & \(\sqrt{\frac{m-1}{m}}\)
    & \multicolumn{2}{c}{Theorem~\ref{thm:flat-main}}\\
\end{tabular}
}
\caption{Types of Quantum Pyramids}
\label{tab:pyramid-types}
\end{table}
We immediately see that ort pyramids have a single state and thus zero accessible information, while orthogonal pyramids are perfectly distinguishable with a computational basis measurement. This leaves three non-trivial regimes: acute, obtuse, and flat. 

Compressing the constraint in \eqref{eq:pyramid-ab} and the conditions in Table~\ref{tab:pyramid-types}, we introduce a single parameter $p$ such that:
\begin{equation}\label{eq:ab-p}
    a^2=p,\quad b^2=q\equiv \frac{1-p}{m-1}
\end{equation}
where we take $b=+\sqrt q$ and $\frac 1 m \leq p\leq 1$ for the acute regime and $b=-\sqrt q $ and $ \frac{m-1}{m}\leq p \leq 1 $ for the obtuse and flat regimes. Using this parameterization and their dual machinery, \citep{holevo2025quantum} prove the following entropy inequality to confirm the optimal measurement conjectured by \citep{englert2010well} for acute pyramids:
\begin{theorem}[Theorem 2 \citep{holevo2025quantum}]\label{thm:acute-ineq}
   For $m\geq 3$, $p\in \left(\frac{m-1}{m},1\right)$, and any $z\in \C^m$ with $\|z\|_2=1$,
   \begin{equation}
       -\sum_{j=1}^m|z_j|^2 \log_2 |z_j|^2 \geq \mu_0(p)\left|\sum_{j=1}^m z_j\right|^2-\mu_1(p)
   \end{equation}
   where 
   \begin{equation}
   \begin{aligned}
       \mu_0(p)&=\frac{\sqrt{pq}(\log_2 p-\log_2 q)}{(\sqrt p+(m-1)\sqrt q)(\sqrt p-\sqrt q)}\\
       \mu_1(p)&=\frac{\sqrt p\log_2 p-\sqrt q \log_2 q }{\sqrt p-\sqrt q }
   \end{aligned}
   \end{equation}
\end{theorem}
\begin{proof}
    Since $\mu_0(p)\geq 0$, $\left|\sum_{j=1}^mz_j\right|^2\leq \left(\sum_{j=1}^m \sqrt{t_j}\right)^2 $ for $t_j\equiv |z_j|^2$. Thus, it suffices to prove the following bound on the probabilities $\{t_j\}$:
    \begin{equation}\label{eq:acute-ineq-t}
        -\sum_{j=1}^m t_j \log_2 t_j -\mu_0(p)\left(\sum_{j=1}^m \sqrt{t_j}\right)^2\geq -\mu_1(p)
    \end{equation}
    Theorem 2 of \citep{holevo2025quantum} does exactly this. Considering the minimizers of the left hand side of \eqref{eq:acute-ineq-t}, \citep{holevo2025quantum} use Lagrange multipliers to show that every extrema must have at most two coordinate values, reducing \eqref{eq:acute-ineq-t} to a soluble single variable inequality in $q$ with one additional discrete multiplicity parameter. 
\end{proof}
By the dual arguments in \citep{holevo2025quantum}, Theorem~\ref{thm:acute-ineq} implies that the following measurement is optimal. 
\begin{corollary}
   Consider the observable family $\c M =\{M_k\}_{k=0}^m$ parameterized by $t\in [0,1]$ such that $M_k=\ketbra{\Tilde e_k}$ where:
   \begin{equation}\label{eq:acute-M}
   \ket{\Tilde e_0}\equiv \sqrt{1-t^2}\ket{e_0},\quad 
       \ket{\Tilde e_k}\equiv \ket{e_k}+\frac{t-1}{\sqrt m}\ket{e_0}\quad (k\in [m])
   \end{equation}
   Then, for acute pyramids, the global information-optimal measurement is in this family with parameter $t=t^*_{A}$:
   \begin{equation}
       t^*_A(m)\equiv \begin{cases}
           1 & p\in \left[\frac{m-1}{m},1\right)\\
           \frac{2(m-1)}{m-2}\cdot \frac{\sqrt p-\sqrt q}{\sqrt p+(m-1)\sqrt q} & p\in \left(\frac 1 m ,\frac{m-1}{m}\right]
       \end{cases}
   \end{equation}
   This yields the accessible information:
   \begin{equation}
       A({\c E})=\begin{cases}
           I_{\rm SRM} & p\in \left[\frac{m-1}{m},1\right)\\
           \frac{m-1}{m-2}\log_2(m-1)(\sqrt p-\sqrt q)^2 &  p\in \left(\frac 1 m ,\frac{m-1}{m}\right]
       \end{cases}
   \end{equation}
   where 
   \begin{equation}
       I_{\rm SRM}\equiv \log_2 m  +p\log_2 p+(m-1)q\log_2 q 
   \end{equation}
\end{corollary}
When $t^*(m)=1$, $\c M =\{\ketbra{e_k}\}_{k\in[m]}$ is called the square-root measurement (SRM) since one can write:
\begin{equation}
    \ket{e_j}=\frac 1 {\sqrt m}\bar \rho ^{-1/2}\ket{\psi_j }=\sum_{k=1}^{m}(G^{-1/2})_{kj}\ket{\psi_k}
\end{equation}
where $\bar \rho \equiv \frac 1 m \sum_{j=1}^m \ketbra{\psi_j}$ is the \textit{average state} and $G_{ij}\equiv \braket{\psi_i}{\psi_j}$ is the Gram matrix. This measurement is also optimal for the Bayes discrimination problem for $\c E $ \citep{helstrom1976quantum}. However, in the strongly acute regime, the square-root measurement ceases to be optimal. 

Unfortunately, the obtuse and flat regimes are much more difficult. While we would like to prove the necessary entropy inequality in a similar manner to the proof of Theorem~\ref{thm:acute-ineq} in the acute case, we are unable to pass to probabilities $\{t_j\}$ since $b<0$. This makes analyzing the family of minimizers much more challenging. In fact, \citep{holevo2025quantum} remark that proving these inequalities via direct optimization is challenging and instead give proof suggestions that enable numerical simulations of their claims for $m\leq 50$. 
\subsection{Main Results}
Nevertheless, we show that direct optimization is possible with some persistence. In particular, we prove the following entropy inequality for the obtuse regime. 

\begin{theorem}\label{thm:obtuse-59-67}
   Let $m\geq 3$ and $z\in \C^m $ satisfy $\|z\|_2=1$. If $m\geq 7$ and $p\in \left(\frac{m-1}{m},1\right)$,
   \begin{equation}\label{eq:thm-obtuse-59}
       -\sum_{j=1}^m |z_j|^2 \log_2 |z_j|^2 \geq \Tilde \mu_0(p)\left|\sum_{j=1}^m z_j\right|^2 -\Tilde \mu_1(p)
   \end{equation}
   where 
   \begin{equation}
   \begin{aligned}
       \Tilde \mu_0(p)&=\frac{\sqrt{pq}\left(\log_2 p-\log_2 q\right)}{\left(-\sqrt p+(m-1)\sqrt{q}\right)\left(\sqrt p+\sqrt{q}\right)}\\
       \Tilde \mu_1(p)&=\frac{\sqrt p\log_2 p+\sqrt{q}\log_2 q}{\sqrt p+\sqrt{q}}
   \end{aligned}
   \end{equation}
   where equality is attained by $(-\sqrt{q},\cdots,-\sqrt{q},\sqrt p)$, up to permutation and global phase. If $m\leq 6$, let $\tau_o(m)$ be the unique root in $[0,1]$ of the function:
   \begin{equation}
       g(\tau)\equiv \log_2 \frac{m(m-1+\tau^2)}{2}-\frac{m-1+\tau}{m}\log_2(m-1+\tau)^2-\frac{1-\tau}{m}\log_2(1-\tau)^2
   \end{equation}
   and define:
   \begin{equation}
       p(m)\equiv \frac{(m-1+\tau_o(m))^2}{m(m-1+\tau_o(m)^2)}
   \end{equation}
   Then, \eqref{eq:thm-obtuse-59} also holds for $p\in [p(m),1)$ with the same equality family for $p\in (p(m),1)$ and the additional family $(2^{-1/2},-2^{-1/2},0,\dots,0)$, up to permutation and global phase, for $p=p(m)$. In particular, we have the following fixed inequality:
   \begin{equation}\label{eq:thm-obtuse-67}
       -\sum_{j=1}^m |z_j|^2 \log_2 |z_j|^2 \geq \frac{(\tau_o(m)^{-1}-1)}{m}\log_2 \frac{2(\tau_o(m)-1)^2}{m(m-1+\tau_o(m)^2)}\left|\sum_{j=1}^mz_j\right|^2+1
   \end{equation}
   where equality is attained by both $(2^{-1/2},-2^{-1/2},0,\dots,0)$ and $(-\sqrt{q'},\cdots,-\sqrt{q'},\sqrt{p'})$ with $p'\equiv p(m)$ and $q'\equiv \frac{1-p(m)}{m-1}$, up to permutation and global phase.
\end{theorem}
Note that the sign differences in the definition of $\Tilde \mu_0,\Tilde \mu_1$ for the obtuse case versus $\mu_0,\mu_1$ in the acute case are due to the change from $b=-\sqrt q $ to $b=+\sqrt q$ respectively. As before, the dual argument of \citep{holevo2025quantum} then implies the following optimal measurement.
\begin{corollary}\label{corr:obtuse-info}
   For $m\geq 3$, consider the observable family $\c M=\{M_k\}_{k\in [m]}\cup \{M_{ij}\}_{1\leq i<j\leq m}$ parameterized by $t\geq 1$ such that:
   \begin{equation}\label{eq:M-k-ij}
       M_k\equiv \frac 1 {t^2}\ketbra{\Tilde e_k},\quad M_{ij}\equiv \frac 2 m \left(1-\frac 1 {t^2}\right)\ketbra{ij}\quad (k\in[m],1\leq i<j\leq m)
   \end{equation}
   where $ \ket{ij}\equiv \frac 1 {\sqrt 2}(\ket{e_i}-\ket{e_j})$ and $\ket{\Tilde e_k}$ are given in \eqref{eq:acute-M}. Then, for obtuse pyramids, the global information-optimal measurement is in this family with parameter $t=t^*_{O}$:
   \begin{equation}
       t^*_O(m)\equiv \begin{cases}
          \tau_o(m)\cdot \frac{\sqrt p+\sqrt q}{\sqrt p-(m-1)\sqrt q} & 3\leq m \leq 6\text{ and }p\in \left(\frac{m-1}{m},p(m)\right]\\
          1 & \text{o.w.}
       \end{cases}
   \end{equation}
   where $\tau_o(m)$ is defined in Theorem~\ref{thm:obtuse-59-67}. This yields the accessible information:
   \begin{equation}
       A(\mathcal E )=\begin{cases}
       \log_2 (m/2)-(\sqrt p-(m-1)\sqrt q)^2\Tilde \mu_0(p(m)) & 3\leq m \leq 6 \text{ and } p\in \left(\frac{m-1}{m},p(m)\right]\\
       I_{\rm SRM} & \text{o.w.}
       \end{cases}
   \end{equation}
\end{corollary}
Indeed, the lifted trine of \citep{shor2002number}, for which the square-root measurement is sub-optimal, is a nearly-flat obtuse pyramid with $m=3$ and $0<r_0<\bar r_0\approx 0.0614$. By \eqref{eq:pyramid-r01}, \eqref{eq:pyramid-ab}, \eqref{eq:ab-p}, it follows that $p=\frac 1 m \left(\sqrt r_0+(m-1)\sqrt{\frac{1-r_0}{m-1}}\right)^2$, meaning the corresponding threshold for $p$ is $\approx 0.8725$, which is exactly the value of $p(3)$ computed numerically in \citep{holevo2025quantum}.

While flat pyramids are a limiting case of the obtuse regime, the standard dual arguments of \citep{holevo2025quantum} degenerate at this limit. Fortunately, \citep{holevo2025quantum} showed that it suffices to establish the following tight entropy inequality to certify optimality. 
\begin{theorem}\label{thm:flat-ineq}
Let $m\geq 3$ and $z\in \C^m $ with $\|z\|_2=1$ and $S(z)=0$, where $S(z)\equiv \sum_{j=1}^m z_j$. Then, 
\begin{equation}\label{eq:flat-ineq}
    -\sum_{j=1}^m |z_j|^2 \log_2 |z_j|^2 \geq \begin{cases}
        1 & m \leq 6\\
        \log_2 m -\frac{m-2}{m}\log_2(m-1) & m \geq 7
    \end{cases}
\end{equation}    
\end{theorem}

In particular, we prove a generalized tight $\ell^p$ inequality on zero-sum vectors that was subsequently conjectured by \citep{holevo2026conjecturetightnorminequality}. 
\begin{theorem}\label{thm:flat-main}
    For $\alpha\in (0,1)\cup(1,\infty)$, let $d(\alpha)$ be the largest root of $2^{1-\alpha}=d^{-\alpha}((d-1)^\alpha+(d-1)^{1-\alpha})$, and define constant 
    $$M(d,\alpha)=\begin{cases}
        2^{1-\alpha}&\alpha \leq 1/2\\
        2^{1-\alpha}& \alpha > 1/2\text{ and }d\leq d(\alpha)\\
        d^{-\alpha}((d-1)^\alpha+(d-1)^{1-\alpha}) & \alpha > 1/2 \text{ and }d\geq d(\alpha) 
    \end{cases}$$
    Then, letting $\c F_d =\{x\in \R^d: \|x\|_2=1\text{ and } \sum_{j=1}^{d}x_j=0\}$,
    $$M(d,\alpha)=\begin{cases}
        \min_{x\in \c F_d}\|x\|_{2\alpha}^{2\alpha}& \alpha < 1\\
        \max_{x\in \c F_d}\|x\|_{2\alpha}^{2\alpha}& \alpha > 1
    \end{cases}$$
    where the branch $M(d,\alpha)=2^{1-\alpha}$ is saturated by the family $(2^{-1/2},-2^{-1/2},0,\dots,0)$ up to permutation and negation, while the branch $M(d,\alpha)=d^{-\alpha}((d-1)^{\alpha}+(d-1)^{1-\alpha})$ is saturated by the family 
    $(((d-1)/d)^{1/2},-(d(d-1))^{-1/2},\dots,-(d(d-1))^{-1/2})$ up to permutation and negation. These are the only families which achieve the optimum value of $\|x\|_{2\alpha}^{2\alpha}$, unless $(d,\alpha)=(3,2)$ in which case any $x\in \c F $ gives $\|x\|_4^4=1/2$.
\end{theorem}
While \citep{holevo2026conjecturetightnorminequality} prove Theorem~\ref{thm:flat-main} for $d=3$ via trigonometric series expansion and careful analysis, an analytic proof for $d\geq 4$ is elusive, though the claim can be numerically verified for $d\leq 200$. By giving an analytic proof for $d\geq 4$, we consequently confirm the entropy inequality of Theorem~\ref{thm:flat-ineq} via the following reduction.
\begin{proposition}
    Theorem~\ref{thm:flat-main} implies Theorem~\ref{thm:flat-ineq}.
\end{proposition}
\begin{proof}
    Suppose $z \in \R^m$. Since $\|z\|_2=1$, $\{t_j\}_{j\in[m]}$ is a probability distribution for $t_j\equiv |z_j|^2$ as before. Now, consider its $\alpha$-R\'enyi entropy:
    \begin{equation}
        H_\alpha(t)=\frac{1}{1-\alpha}\log_2 \sum_{j=1}^mt_j^\alpha
    \end{equation}
    For $\alpha \in (0,1)$, Theorem~\ref{thm:flat-main} gives that $\sum_{j=1}^{m}t_j^\alpha=\|x\|_{2\alpha}^{2\alpha}\geq M(m,\alpha)$. Then, 
    \begin{equation}
        H_\alpha(t)\geq \frac 1 {1-\alpha}\log_2 M(m,\alpha)\xrightarrow{\alpha\to1^-} -\sum_{j=1}^m t_j \log_2 t_j \geq \lim_{\alpha \to 1^-}\frac{\log_2 M(m,\alpha)}{1-\alpha}
    \end{equation}
    By \citep{holevo2026conjecturetightnorminequality}, $d(\alpha)\in (6,7)$ for $\alpha$ sufficiently close to $1$. Thus, for $m\leq 6$, $M(m,\alpha)= 2^{1-\alpha}$ which proves this branch of Theorem~\ref{thm:flat-ineq}. For $m\geq 7$, by L'H\^{o}pital's rule,
    \begin{equation}
    \begin{aligned}
        \lim_{\alpha \to 1^-}\frac{\log_2 M(m,\alpha)}{1-\alpha}&=-\frac d {d\alpha}\log_2 (m^{-\alpha}((m-1)^\alpha+(m-1)^{1-\alpha}))\bigg|_{\alpha=1}\\
        &=\log_2 m-\frac{(m-1)^\alpha-(m-1)^{1-\alpha}}{(m-1)^\alpha+(m-1)^{1-\alpha}}\log_2(m-1)\bigg|_{=1}\\
        &=\log_2 m -\frac{m-2}{m}\log_2(m-1)
    \end{aligned}
    \end{equation}
    as desired, proving Theorem~\ref{thm:flat-ineq} for $z\in \R^m$. We now argue that this gives the claim for $z\in \C^m$. Let $z=x+iy$ for $x,y\in \R^m$, and let $C_m$ denote the right hand side of \eqref{eq:flat-ineq}. Letting $\lambda\equiv \|x\|_2^2\in (0,1)$, we can write:
    \begin{equation}
        |z_j|^2=\lambda \frac{x_j^2}{\lambda}+(1-\lambda)\frac{y_j^2}{1-\lambda}
    \end{equation}
    where $\{x_j^2/\lambda\}_{j\in [m]}$ and $\{y_j^2/(1-\lambda)\}_{j\in[m]}$ are probability distributions. Then, by the concavity of Shannon entropy, 
    \begin{equation}
        -\sum_{j=1}^m|z_j|^2 \log_2 |z_j|^2 \geq -\lambda \sum_{j=1}^m \frac{x_j^2}{\lambda^2}\log_2 \frac{x_j^2}{\lambda^2}-(1-\lambda) \sum_{j=1}^m  \frac{y_j^2}{\lambda^2}\log_2 \frac{y_j^2}{\lambda^2}\geq \lambda C_m+(1-\lambda) C_m=C_m
    \end{equation}
    where we use the fact that $\frac{x}{\sqrt \lambda},\frac{y}{\sqrt{1-\lambda}}\in \R^m $, have unit $\ell^2$ norm, and have zero sum, meaning Theorem~\ref{thm:flat-ineq} for real variables applies. Of course, if $\lambda \in \{0,1\}$, this is exactly the real case. This argument yields the proof for complex variables and is a simple specialization of a broader real reduction trick given in Lemma 1 of \citep{holevo2025quantum}, which we state as Lemma~\ref{lemma:C-to-R} in our later proof of the obtuse case.  
\end{proof}
By \citep{holevo2025quantum}, Theorem~\ref{thm:flat-ineq} yields the following optimal measurement.
\begin{corollary}
    For flat pyramids, the global information-optimal measurement is in the family defined in Corollary~\ref{corr:obtuse-info} with parameter $t=t^*_F$:
    \begin{equation}
        t^*_F(m)\equiv \begin{cases}
            \infty & 3\leq m \leq 6\\
            1 & m \geq 7
        \end{cases}
    \end{equation}
    This yields the accessible information:
    \begin{equation}
        A({\c E })=\begin{cases}
            0 & m=2\\
            \log_2(m/2) & 3\leq m \leq 6\\
            \frac{m-2}{m}\log_2(m-1) & m \geq 7
        \end{cases}
    \end{equation}
    where we note that ${\c E}$ degenerates to a single state up to phase when $m=2$.
\end{corollary}
From Table~\ref{tab:pyramid-types} and \eqref{eq:pyramid-ab}, flat pyramids lie in the hyperplane $\{\psi:\braket{e_0}{\psi}=0\}$. Then, \eqref{eq:acute-M} and \eqref{eq:M-k-ij} for $m\leq 6$ effectively give the optimal pair-difference measurement ${\c M}=\{\ketbra{e_0},M_{ij}\}_{1\leq i<j\leq m}$ where $M_{ij}=\frac 2 m \ketbra{ij}$. For $m=3$, this is exactly the \textit{anti-trine} measurement from \citep{weir2018optimal}. 

\subsection{Organization}
Our proofs of Theorem~\ref{thm:obtuse-59-67} and Theorem~\ref{thm:flat-main} employ a similar strategy to the numerical verification algorithms of \citep{holevo2025quantum,holevo2026conjecturetightnorminequality}. We analyze the minimizers of the entropy inequality, progressively refining the minimizer family until we can compress it into a tractable inequality of a single continuous variable and at most one discrete parameter. We then patiently solve these optimization problems. Section~\ref{sec:obtuse} proves Theorem~\ref{thm:obtuse-59-67} for the obtuse case and importantly shows that the elimination of the hard family, that is minimizers with three distinct values, reduces to a neat algebraic reciprocal inequality involving the branches of the Lambert $W$ function which we state and prove in Appendix~\ref{x:obtuse}. Section~\ref{sec:flat} proves Theorem~\ref{thm:flat-main} related to the flat case, taking advantage of a technique in symmetric inequalities known as the equal variables method \citep{cirtoaje2006algebraic}.

\section{Obtuse Pyramids}\label{sec:obtuse}
We roughly follow the proof strategy suggested in Section 5.2 of \citep{holevo2025quantum}. First, we employ Lemma 1~\citep{holevo2025quantum} to reduce from complex to real variables. Then, we leverage Lagrange multipliers to argue that the relevant minimizers have at most three distinct non-zero values, which can be further assumed to take a particular shape. Our key contribution is a proof that no minimizers with three distinct values exist. With this fact, analyzing the remaining minimizer families proves Theorem~\ref{thm:obtuse-59-67} or equivalently Equations (59) and (67) of \citep{holevo2025quantum}.

\subsection{Preliminary Reduction}
We begin by reducing to real variables via \citep{holevo2025quantum}'s proof using the concavity of Shannon entropy.  
\begin{lemma}[Real Reduction, Lemma 1 of \citep{holevo2025quantum}]\label{lemma:C-to-R}
If $\lambda_{jk}\in \R $ and 
\begin{equation}
    \sum_{j=1}^m |z_j|^2 \log_2 |z_j|^2 +\sum_{j,k=1}^m\lambda_{jk}\bar{z}_j z_k \leq 0
\end{equation}
holds for all $z\in \R^m  $ such that $\|z\|_2=1$, then it also holds for all $z\in \C^m $ such that $\|z\|_2=1$. 
\end{lemma}
Taking $\lambda_{jk}={\Tilde \nu_1}\delta_{jk}-{\Tilde \nu_0}$, it suffices to show that for $x\in \S^{m-1}$, 
\begin{equation}\label{eq:obtuse-fc0}
   f_{{\Tilde \nu_0}}(x)\equiv -\sum_{j=1}^m x_j^2 \ln x_j^2 +{\Tilde \nu_0}S(x)^2 \geq {\Tilde \nu_1} 
\end{equation}
where $S(x)\equiv \sum_{j=1}^m x_j$ and $\Tilde\nu_i(p)\equiv -\Tilde \mu_i(p)\ln 2 $ such that ${\Tilde \nu_0}>0$ in the obtuse regime for $p$. Note that for the remainder of the proof, we will work with natural logarithms. We now begin our quest to construct and handle the relevant families of minimizers.
\begin{lemma}\label{lemma:obtuse-lagrange}
    Define the Lagrange function of \eqref{eq:obtuse-fc0} along with its respective gradient and Hessian:
    \begin{equation}\label{eq:obtuse-lagrange}
    \begin{aligned}
         \mathcal L_{{\Tilde \nu_0}} (x;\lambda)&\equiv f_{{\Tilde \nu_0}} (x)-\lambda \left(\sum_{j=1}^m x_j^2 -1\right)\\
       \p_j \mathcal{L}_{{\Tilde \nu_0}}(x;\lambda)&=-2x_j(\ln x_j^2 + 1+\lambda)+2{\Tilde \nu_0}S(x)\\
        \p_i\p_j  \mathcal L (x;\lambda)&= -2(\ln x_j^2+3+\lambda)\delta_{ij}+2{\Tilde \nu_0} 
    \end{aligned}
    \end{equation}
    If $x\in \S^{m-1}$ is a local minimizer of $f_{{\Tilde \nu_0}}(x)$, there exists a $\lambda \in \R$ such that $\nabla_x\mathcal{L}_{{\Tilde \nu_0}}(x;\lambda)=0$ and $y^\top \nabla^2_x \mathcal{L}_{{\Tilde \nu_0}}(x;\lambda)y\geq 0$ for any $y\in \R^m$ with $\langle x,y\rangle=0$. 
\end{lemma}
\begin{proof}
    This follows from direct computation and standard calculus.
\end{proof}
We now use these Lagrangian conditions to narrow the search space for minimizers. 

\begin{lemma}\label{lemma:obtuse-zero}
   If $m\geq 3$, any global minimizer $x\in \S^{m-1}$ of $f_{{\Tilde \nu_0}}(x)$ such that at least one coordinate of $x$ is zero must have the form $x=(2^{-1/2},-2^{-1/2},0,\dots,0)$ up to permutation and negation, achieving $f_{{\Tilde \nu_0}}(x)=\ln 2 $. 
\end{lemma}
\begin{proof}
    Let the zero coordinate be $x_k=0$. Then, $\langle x,e_k\rangle=0$ meaning the unit vector $e_k$ is tangent to $\S^{m-1}$ and by first-order optimality, 
    \begin{equation}
        0=\nabla_{e_k}f_{{\Tilde \nu_0}}(x)=-\cancel{2x_k (\ln x_k^2 +1)}+2{\Tilde \nu_0} S(x)
    \end{equation}
    where we use the convention $0\ln 0\equiv 0$. Then, Lemma~\ref{lemma:obtuse-lagrange} gives:
    \begin{equation}
        0={\Tilde \nu_0} S(x)=x_j(\ln x_j^2 + 1+\lambda)
    \end{equation}
    meaning every $x_j\in \{0,\pm s\}$ for some $s>0$. This along with the fact that $S(x)=0$ implies that there are the same number of positive and negative coordinates, say $r$ of each. Then, $1=\|x\|_2^2=2rs^2$, meaning:
    \begin{equation}
        f_{{\Tilde \nu_0}}(x)=-2r\cdot s^2\ln s^2 +{\Tilde \nu_0} \cdot 0^2 =\ln(2r)
    \end{equation}
    Since we seek minimizers, we have that $r=1$, forcing $s=2^{-1/2}$ and $m\geq 3 $ as well. 
\end{proof}

We now rigorously prove \citep{holevo2025quantum}'s claim that minimizers have at most three distinct values. 
\begin{lemma}\label{lemma:obtuse-three}
   Let $x\in \S^{m-1}$ be an extrema of $f_{{\Tilde \nu_0}}(x)$ with no zero coordinate. Then, we have that $|\{x_1,\dots,x_m\}|\leq 3$. Further, if $S(x)\geq 0$, there is at most one distinct positive value, and if $S(x)\leq 0$, there is at most one distinct negative value. 
\end{lemma}
\begin{proof}
    By Lemma~\ref{lemma:obtuse-lagrange}, every $x_j$ is a non-zero root of $g(t)\equiv t\ln t^2 + (1+\lambda )t +b=0$ where $b\equiv -{\Tilde \nu_0}S(x)$. Taking derivatives, $g'(t)=\ln t^2 +\lambda+3$ and $g''(t)=2/t$, meaning $g$ is strictly convex on $(0,\infty)$, strictly concave on $(-\infty,0)$, with $g'$ having at most one root on each domain, meaning $g$ has at most two roots on each domain. We now refine this using the limiting behavior $\lim_{t\to \pm \infty}g(t)=\pm \infty$ and $\lim_{t\to 0^\pm}=b$. 

    Suppose that $b\leq 0$ and that there are two positive roots $0<s<t$. By Rolle's theorem, there is a $c\in (s,t)$ with $g'(c)=0$, and since $g''>0$ on $(0,\infty)$, $g'<0$ on $(0,c)$, meaning $g(s)<0$ since $b\leq 0$. This contradicts the fact that $s$ is a root, meaning there is at most one positive root if $b\leq 0$. By a similar argument, if $b\geq 0$, there is at most one negative root. 
\end{proof}

We continue refining this three value family by considering the signs of multiplicities of each value. 
\begin{lemma}\label{lemma:obtuse-shape}
    For $m\geq 2 $ and ${\Tilde \nu_0}>0$, any local minimizer $x\in \S^{m-1}$ of $f_{{\Tilde \nu_0}}(x)$ with no zero coordinate must have the form $x=(-s_0,\dots,-s_0,-s_1,\cdots,-s_1,s_2)$ for $0<s_0\leq s_1$ and $s_2>0$ with $S(x)\geq 0$, up to permutation and negation. 
\end{lemma}
\begin{proof}
We first show that $x$ has at least one positive coordinate. Suppose every $x_j<0$. Then, since $S(x)<0$, by Lemma~\ref{lemma:obtuse-three}, there is at most one negative root, meaning $x_j=t=-m^{-1/2}$ since $\|x\|_2=1$ as well. Since $m\geq 2$, there exists a $y\in \R^m$ such that $\langle x,y\rangle =0$ and $S(y)=0$. Then, by Lemma~\ref{lemma:obtuse-lagrange}, $\ln t^2 + 1+\lambda ={\Tilde \nu_0}S(x)/t={\Tilde \nu_0}m$ meaning:
\begin{equation}
        \frac 1 2 y^\top \nabla^2_x \mathcal L _{{\Tilde \nu_0}}(x;\lambda)y=-(\ln t^2+3+\lambda)\|y\|_2^2 =-({\Tilde \nu_0}m+2)\|y\|_2^2<0
\end{equation}
which contradicts the Hessian condition in Lemma~\ref{lemma:obtuse-lagrange}. Thus, $x$ has at least one positive coordinate, and by negation, this proves it must have at least one negative coordinate as well. 

WLOG, by negation, let $S(x)\geq 0$. We now show that no positive value can occur in two or more coordinates of such $x$. For example, suppose $x_i=x_k=s>0$ for indices $i\neq k$. Then, defining $y\equiv e_i-e_k $, $\langle x,y\rangle=0$ and $S(y)=0$. Then, by Lemma~\ref{lemma:obtuse-lagrange}, $\ln s^2 +1+\lambda ={\Tilde \nu_0} S(x)/s$, meaning:
\begin{equation}
     \frac 1 2 y^\top \nabla^2_x \mathcal L _{{\Tilde \nu_0}}(x;\lambda)y=-2(\ln s^2+3+\lambda)=-2\left(\frac{{\Tilde \nu_0} S(x)}{s}+2\right)<0
\end{equation}
which contradicts the Hessian condition in Lemma~\ref{lemma:obtuse-lagrange}. 

Together with Lemma~\ref{lemma:obtuse-three}, these facts prove the claim. 
\end{proof}

Towards the hypothesis by \citep{holevo2025quantum} that minimizers have at most two distinct values, we focus on families of three distinct values. 
\begin{lemma}\label{lemma:obtuse-shape-final}
   Suppose $x\in \S^{m-1}$ is a local minima of $f_{{\Tilde \nu_0}}(x)$ with no zero coordinate. Then, if $x$ has three distinct values, $x=(-s_0,\dots,-s_0,-s_1,s_2)$ up to permutation and negation, and if $x$ instead has two distinct values, $x=(-s_0,\dots,-s_0,s_2)$ up to permutation and negation, where $0<s_0<s_1$ and $s_2>0$.
\end{lemma}
\begin{proof}
   It suffices to show that if $x$ has three distinct values, $v$ in Lemma~\ref{lemma:obtuse-shape} has multiplicity $1$.  
   Define $r\equiv s_1/s_0$ and $c\equiv {\Tilde \nu_0} S(x)/s_0$. Then, Lemma~\ref{lemma:obtuse-lagrange} gives $\ln s_0^2 +1 +\lambda =-c$ and $\ln s_1^2 +1+\lambda =-c/r$. Subtracting and re-arranging, 
   \begin{equation}\label{eq:obtuse-c(r)}
       c=\frac{2r\ln r }{r-1}
   \end{equation}
   Now, suppose $s_1$ has multiplicity at least two, meaning there exist distinct indices $i,k\in [m]$ such that $x_i=x_k=-s_1$. Then, let $y=e_i-e_k$ which satisfies $\langle x,y\rangle=0$ and $S(y)=0$. Then, 
   \begin{equation}
        \frac 1 2 y^\top \nabla^2_x \mathcal L _{{\Tilde \nu_0}}(x;\lambda)y=-2(\ln s_1^2 + 3+\lambda)=2(c/r-2)=4\left(\frac{\ln r }{r-1}-1\right)<0
   \end{equation}
   since $\ln r < r-1$ for $r>1$. This contradicts the Hessian condition in Lemma~\ref{lemma:obtuse-lagrange}, meaning $v$ can not have multiplicity at least two. 
\end{proof}

\subsection{Elimination of Three-Value Minimizers}
We now show that the three distinct value family in Lemma~\ref{lemma:obtuse-shape-final} does not admit any minimizers. 
\begin{lemma}\label{lemma:obtuse-no-three}
    There does not exist an $x\in \S^{m-1}$ with no zero coordinate and three distinct coordinate values such that $x$ is a local minimizer of $f_{{\Tilde \nu_0}}(x)$.
\end{lemma}
\begin{proof}
    
We do this by first parameterizing this family as $x=s_0\cdot (-1,\dots,-1,-r,s)$ where $r\equiv s_1/s_0>1$ as before, $s\equiv s_2/s_0>0$, and $k\equiv m-2\geq 1$. 

Since $\|x\|_2=1$, we have that $s_0=(k+r^2+s^2)^{-1/2}$, and $S(x)=s_0(s-k-r)$ by summation. Taking $c\equiv {\Tilde \nu_0} S(x)/s_0$ as in the proof of Lemma~\ref{lemma:obtuse-shape-final}, \eqref{eq:obtuse-c(r)} gives:
\begin{equation}\label{eq:obtuse-c0(r,s)}
    {\Tilde \nu_0}=\frac{c}{s-k-r}=\frac{2r\ln r }{(r-1)(s-k-r)}
\end{equation}
Now, \eqref{eq:obtuse-c(r)} for $r>1$ gives $c>0$, meaning $s>k+r$ since $c,{\Tilde \nu_0},s_0>0$ implies $S(x)>0$. Applying Lemma~\ref{lemma:obtuse-lagrange}, we have $\ln s_0^2 + 1+\lambda=-c$ and $\ln s_2^2+1+\lambda=c/s$. Subtracting, 
\begin{equation}\label{eq:obtuse-s(r)}
    \frac{s\ln s}{1+s}=\frac c 2 =\frac{r\ln r }{r-1}
\end{equation}
where we equate with \eqref{eq:obtuse-c(r)}. Thus, our family is defined by a single parameter $r>1$ where $s$ is given by \eqref{eq:obtuse-s(r)}, ${\Tilde \nu_0}$ is forced by \eqref{eq:obtuse-c0(r,s)}, and $s>k+r$ is necessary so that $S(x),{\Tilde \nu_0}>0$. 

At this point, we seek to construct a $y\in \R^m $ such that $\langle x,y\rangle=0$ with $y^\top \nabla^2_x \mathcal L_{{\Tilde \nu_0}}(x;\lambda)y<0$ and $\nabla \mathcal L_{{\Tilde \nu_0}}(x;\lambda)=0$, the existence of which would exclude all minimizers from the three value family via Lemma~\ref{lemma:obtuse-lagrange}. In particular, we show that there exists such a $y$ that is uniform on each coordinate group. Let $y=(\alpha,\dots,\alpha,\beta,\gamma)\in \R^m $ such that:
\begin{equation}\label{eq:obtuse-gamma-elim}
    0=\langle x,y\rangle =-ks_0\alpha-rs_0\beta+ss_0\gamma \Rightarrow \gamma=\frac{k\alpha+r\beta}{s}
\end{equation}
Recall that $\nabla \mathcal L_{{\Tilde \nu_0}}(x;\lambda)=0$ gives:
\begin{equation}\label{eq:obtuse-three-gradient-zero}
    \ln s_0^2 + 1+\lambda=-c,\quad \ln s_1^2 +1 +\lambda =-c/r,\quad \ln s_2^2+1+\lambda=c/s
\end{equation}
Noting that the diagonal term in the Hessian expression in \eqref{eq:obtuse-lagrange} is the left hand side of \eqref{eq:obtuse-three-gradient-zero} up to an additive constant of $+2$, 
\begin{equation}\label{eq:obtuse-hessian-M}
\begin{aligned}
    y^\top \nabla_x^2 \mathcal L_{{\Tilde \nu_0}}(x;\lambda)y&=-2(k\alpha^2 (2-c)+\beta^2 (2-c/r)+\gamma^2 (2+c/s))+2{\Tilde \nu_0}(k\alpha+\beta +\gamma)^2\\
    &=M_{11}\alpha^2+2M_{12}\alpha\beta +M_{22}\beta^2 
\end{aligned}
\end{equation}
where substituting \eqref{eq:obtuse-gamma-elim} and grouping terms gives:
\begin{align}
M_{11}
&\equiv k h_0+\frac{k^2}{s^2}h_2+2{\Tilde \nu_0} k^2\left(1+1/s\right)^2\\
M_{12}
&\equiv \frac{kr}{s^2}h_2+2{\Tilde \nu_0} k\left(1+1/s\right)\left(1+r/s\right)\\
M_{22}
&\equiv h_1+\frac{r^2}{s^2}h_2+2{\Tilde \nu_0}\left(1+r/s\right)^2
\end{align}
with $h_0\equiv -2(2-c)$, $h_1\equiv -2(2-c/r)$, and $h_2\equiv -2(2+c/s)$.
Since \eqref{eq:obtuse-hessian-M} is exactly a quadratic form of the matrix $M$, it suffices to show that $\det M<0$. Substituting \eqref{eq:obtuse-c0(r,s)}, expanding, and factoring, we arrive at the following expression which one can verify by direct expansion:
\begin{equation}
    \det M =M_{11}M_{22}-M_{12}^2 =-\frac{4k(k+r^2+s^2)}{rs^3(s-k-r)}P_{k}(r,s,c)
\end{equation}
where 
\begin{equation}
\begin{aligned}
    P_k(r,s,c)&\equiv-c^2 k-c^2 r^2-c^2 s^2+2ckr-2cks-2cr^2s+2cr^2+2crs^2+2cs^2+4krs+4r^2 s-4rs^2\\
    &=P_1(r,s,c)+(k-1)(2r-c)(c+2s)
\end{aligned}
\end{equation}
Since $s>k+r$, $r>1$, and $k\geq 1$, $\det M <0$ is equivalent to $P_k(r,s,c)>0$. By \eqref{eq:obtuse-c(r)}, $r>1$ implies $\ln r<r-1$ meaning $c>0$ and $2r-c=2r(1-\ln (r)/(r-1))>0$, so $(k-1)(2r-c)(c+2s)>0$. Thus, it suffices to show $P_1(r,s,c)>0$. 

We now re-parameterize to derive a scalar inequality involving parameters restricted to bounded domains. Specifically, let $\sigma=1/r\in (0,1)$ and $t=s/r>1$, though one can show $t\in (1,t_0)$ for some unique $t_0$. Then, \eqref{eq:obtuse-s(r)} becomes:
\begin{equation}\label{eq:obtuse-sigma-t}
    \frac{rt(\ln r+\ln t)}{1+rt}=\frac{r\ln r }{r-1}\Rightarrow t(r-1)\ln t =(1+rt-t(r-1))\ln r \Rightarrow  \frac{t\ln t }{1+t}=\frac{\ln r }{r-1}=-\frac{\sigma \ln \sigma}{1-\sigma }\equiv \kappa
\end{equation}
From \eqref{eq:obtuse-c(r)}, we also have that $c=2\kappa/\sigma$. Substituting this parameterization into $P_1(r,s,c)$, 
\begin{equation}\label{eq:obtuse-K}
\begin{aligned}
    P_1(r,s,c)&=-\frac 4{\sigma^4 }K(\sigma,t)\\
    K(\sigma,t)&\equiv \kappa^2(\sigma^2+t^2+1)+\kappa(\sigma^2(t-1)-\sigma(t^2+1)-t(t-1))+\sigma t(t-\sigma-1)
\end{aligned}
\end{equation}
Thus, it suffices to show that $K(\sigma,t)<0$ on the curve \eqref{eq:obtuse-sigma-t}. 

We now use yet another re-parameterization to reduce this to a simple algebraic inequality. Let $\alpha\equiv \kappa/t$ and $\beta \equiv \kappa/\sigma $. We first claim that $0<\alpha<\kappa<1<\beta$ and $\alpha e^\alpha=\kappa e^{-\kappa}=\beta e^{-\beta}$. Manipulating,
    \begin{equation}
        \kappa=\frac{t\ln t}{1+t}\Rightarrow \alpha=\frac{\ln t }{1+t}\Rightarrow \ln (\kappa/\alpha)=\alpha(1+\kappa/\alpha)\Rightarrow \alpha e^\alpha=\kappa e^{-\kappa}
    \end{equation}
    Similarly,
    \begin{equation}
        \kappa=-\frac{\sigma\ln \sigma}{1-\sigma}\Rightarrow \beta=-\frac{\ln \sigma}{1-\sigma}\Rightarrow -\ln(\kappa/\beta)=\beta(1-\kappa/\beta)\Rightarrow \kappa e^{-\kappa}=\beta e^{-\beta}
    \end{equation}
    Since $\sigma \in (0,1)$, $\beta =-\frac{\ln \sigma}{1-\sigma}>1$, and $\alpha < \kappa$ since $t>1$, which concludes our initial claim. Substituting this parameterization into \eqref{eq:obtuse-K}, one can verify by direct expansion that:
    \begin{equation}
       K(\sigma,t)=K(\kappa/\beta,\kappa/\alpha)=\frac{\kappa^2}{\alpha^2\beta^2} \alpha(1+\alpha)\kappa(1-\kappa)\beta(\beta-1)\left(\frac{1}{\kappa(1-\kappa)}-\frac{1}{\alpha(1+\alpha)}-\frac{1}{\beta(\beta-1)}\right)
    \end{equation}
    Noting that the entire pre-factor is positive, it suffices to show that:
    \begin{equation}
        \frac{1}{\kappa(1-\kappa)}<\frac{1}{\alpha(1+\alpha)}+\frac{1}{\beta(\beta-1)}
    \end{equation}
    We prove in Lemma~\ref{lemma:obtuse-reciprocal} that this is indeed a neat algebraic fact for any $0<\alpha<\kappa<1<\beta$ with $\alpha e^\alpha=\kappa e^{-\kappa}=\beta e^{-\beta}$. 
    
\end{proof}

    This shows that there do not exist minimizers in the three distinct value family, thereby proving the hypothesis from \citep{holevo2025quantum} that minimizers have two distinct values. This finishes the case of obtuse pyramids by Lemma 5 and Theorem 3 of \citep{holevo2025quantum}. However, since we also have the guarantee that the positive value has multiplicity one, our handling of the two value family to prove Theorem~\ref{thm:obtuse-59-67} is much simpler than the treatment provided in Lemma 5 of \citep{holevo2025quantum}.

\begin{lemma}\label{lemma:obtuse-two}
   Suppose $p> (m-1)/m$. If $x\in \S^{m-1}$ is a local minimizer of $f_{{\Tilde \nu_0}(p)}(x)$ with no zero coordinate and two distinct values such that $S(x)\geq 0$, then
   \begin{equation}
       f_{{\Tilde \nu_0}(p)}(x)\geq {\Tilde \nu_1}(p)
   \end{equation}
   with equality attained when $x=(-\sqrt q ,\cdots,-\sqrt q,\sqrt p)$ up to permutation and negation.
\end{lemma}
\begin{proof}
    From Lemma~\ref{lemma:obtuse-shape-final}, $x=(-\sqrt {q'},\cdots,-\sqrt {q'},\sqrt {p'})$ for some $p',q'>0$. Then, $\|x\|_2=1$ implies $q'=(1-p')/(m-1)$, and $S(x)\geq 0$ implies $\sqrt{p'}\geq (m-1)\sqrt{q'}$, meaning $p'\geq (m-1)/m$. Letting $r'\equiv \sqrt{p'/q'}\geq m-1$, we have:
    \begin{equation}\label{eq:obtuse-pqr}
        p'=\frac{r'^2}{r'^2+m-1},\quad  q'=\frac{1}{r'^2+m-1}
    \end{equation}
    meaning:
    \begin{equation}
    \begin{aligned}
        -\sum_{j=1}^m x_j^2 \ln x_j^2 &=-p'\ln p' -(m-1)q'\ln q' =\ln (r'^2+m-1)-\frac{2r'^2\ln r'}{r'^2+m-1}\equiv H_m(r')\\
        S(x)^2&=(\sqrt {p'}-(m-1)\sqrt {q'})^2=\frac{(r'-m+1)^2}{r'^2+m-1}\equiv R_m(r')
    \end{aligned}
    \end{equation}
    We now argue that $H_m$ as a function of $R_m$ is convex. Differentiating, one can verify that:
    \begin{equation}
     H_m'(r')=-\frac{4(m-1)r' \ln r' }{(m-1+r'^2)^2},\quad 
          R_m'(r')=\frac{2(m-1)(r'-m+1)(r'+1)}{(m-1+r'^2)^2}>0
          \end{equation}
          Then, one can also verify that:
          \begin{align}\label{eq:obtuse-two-HmRm-slope}
       \frac{dH_m}{dR_m}&=-\frac{2r' \ln r' }{(r'-m+1)(r'+1)}\\
          \frac{d}{dr'}\left(\frac{dH_m}{dR_m}\right)&=\frac{2N_m(r')}{(r'-m+1)^2(r'+1)^2}
    \end{align}
    where $N_m(r')\equiv (m-1+r'^2)\ln r' -(r'-m+1)(r'+1)$. We now show that $r' \geq m-1$ implies this numerator is non-negative. 
    \begin{equation}
        \begin{aligned}
            N_m(m-1)&=m(m-1)\ln (m-1) \geq 0\\
            N_m'(m-1)&=2r' \ln r' +\frac{m-1+r'^2}{r' }-(r'+1)-(r'-m+1)\bigg|_{r'=m-1}=2(m-1)\ln (m-1)\geq 0\\
            N_m''(r')&=2\ln r' +1-\frac{m-1}{r'^2}\geq 0\quad
        \end{aligned}
    \end{equation}
    Thus, $N_m'(r')\geq 0$ for $r' \geq m-1$, and $N_m(r')\geq 0$ for $r' \geq m-1$ as well. Since $R_m'(r')>0$ and $\frac d {dr'}(dH_m/dR_m)\geq 0$, it follows that $dH_m/dR_m$ is increasing as a function of $R_m$, and is thus convex in $R_m$. 
    
    Applying convexity at the point $r'=r$, by the slope in \eqref{eq:obtuse-two-HmRm-slope}, 
    \begin{equation}\label{eq:obtuse-two-convex-at-r}
        H_m(r')\geq H_m(r)+\frac{dH_m}{d R_m}\bigg|_{r'=r}(R_m(r')-R_m(r))=H_m(r)-\Tilde \nu_1(p)(R_m(r')-R_m(r))
    \end{equation}
    where one can verify the last equality by substituting $r'=r\equiv \sqrt{p/q}$ into \eqref{eq:obtuse-two-HmRm-slope}. Re-arranging \eqref{eq:obtuse-two-convex-at-r}, it follows that:
    \begin{equation}\label{eq:obtuse-two-convex-bound}
    \begin{aligned}
        f_{{\Tilde \nu_0}(p)}(x)&=-\sum_{j=1}^mx_j^2\ln x_j^2 +{\Tilde \nu_0}(p)S(x)^2=H_m(r')+{\Tilde \nu_0}(p)R_m(r')\\
        &\geq H_m(r)+{\Tilde \nu_0}(p)R_m(r)\\
        &=-p\ln p -(m-1)q\ln q +\frac{\sqrt{pq}(\ln p-\ln q)}{(\sqrt p-(m-1)\sqrt q)(\sqrt p+\sqrt q)}(\sqrt p-(m-1)\sqrt q)^2\\
        &={\Tilde \nu_1}(p)
    \end{aligned}
    \end{equation}
    where one can verify the last equality. It is immediate that equality exclusively occurs for $r'=r$. 
\end{proof}

\subsection{Proof of Theorem~\ref{thm:obtuse-59-67}}
We now prove Theorem~\ref{thm:obtuse-59-67} or equivalently Equations (59) and (67) in \citep{holevo2025quantum}. 
\begin{proof}
    From Lemmas~\ref{lemma:obtuse-zero}, \ref{lemma:obtuse-shape-final}, \ref{lemma:obtuse-no-three}, and \ref{lemma:obtuse-two}, we have that the two relevant competing local minima are $(2^{-1/2},-2^{-1/2},0,\cdots,0)$, which attains $f_{{\Tilde \nu_0}(p)}(x)=\ln 2 $, and $(-\sqrt q,\cdots,-\sqrt q,\sqrt p)$, which gives $f_{{\Tilde \nu_0}(p)}(x)= {\Tilde \nu_1}(p)$ for any $p> (m-1)/m$. From \eqref{eq:obtuse-two-convex-at-r} and \eqref{eq:obtuse-two-convex-bound}, 
    \begin{equation}\label{eq:thm-obtuse-59-c1}
    \begin{aligned}
        {\Tilde \nu_1}(r)&=f_{{\Tilde \nu_0}(r)}(x)=H_m(r)+{\Tilde \nu_0}(r) R_m(r)=H_m(r)-\frac{dH_m}{dR_m}(r)R_m(r)\\
        {\Tilde \nu_1}'(r)&={\Tilde \nu_0}'(r)R_m(r)
    \end{aligned}
    \end{equation}
    However, since the proof of Lemma~\ref{lemma:obtuse-two} implies that $R_m(r)>0$ and $\frac{d}{dr}(\frac{dH_m}{dR_m}(r))> 0$ for $p>(m-1)/m$, we have that ${\Tilde \nu_1}$ is decreasing in $r$, meaning it is decreasing in $p$, since the analogous relation of \eqref{eq:obtuse-pqr} for $p$ and $r$ implies that $p$ increases with $r$. Now, from Equation (66) of \citep{holevo2025quantum}, we have that ${\Tilde \nu_1}(p(m))=\ln 2$. This immediately gives the claim about \eqref{eq:thm-obtuse-59} for $m\leq 6$. For $m\geq 7$, taking $p\to^+(m-1)/m$ or equivalently $r\to^+m-1$, from \eqref{eq:thm-obtuse-59-c1},
    \begin{equation}
        \lim_{r\to (m-1)^+}{\Tilde \nu_1}(r)\leq H_m(m-1)=\ln m -\frac{m-2}{m}\ln(m-1)\equiv \ln 2-g(m)
    \end{equation}
    Differentiating,
    \begin{equation}
        g'(m)=\frac{2(m-1)\ln(m-1)-m}{m^2(m-1)}>0
    \end{equation}
    since the numerator is clearly positive and increasing for $m\geq 7$. Since $g(7)>0$ as well by computation, $g(m)>0$. Since ${\Tilde \nu_1}(r)$ is decreasing, it follows that ${\Tilde \nu_1}(p)<\ln 2 $ for all $p\in ((m-1)/m,1]$ with $m\geq 7$, proving the claim about \eqref{eq:thm-obtuse-59} for $m\geq 7$. 

    Simplifying \eqref{eq:thm-obtuse-59} for $p=p(m)$ with $m\leq 6$, we precisely obtain \eqref{eq:thm-obtuse-67} and its claimed minimizers. 
\end{proof}

\section{Flat Pyramids}\label{sec:flat}
Inspired by the numerical proof strategy of \citep{holevo2026conjecturetightnorminequality} for $d\leq 200$, we prove Theorem~\ref{thm:flat-main} by carefully analyzing various families of minimizers, reducing to scalar inequalities when possible. Our key tool is the equal variables method from the mathematical literature on symmetric inequalities \citep{cirtoaje2006algebraic}, but before getting ahead of ourselves, we start with a few simple regimes.

\subsection{Preliminary Reduction}
Verifying that the stated families in Theorem~\ref{thm:flat-main} indeed achieve $\|x\|_{2\alpha}^{2\alpha}=M(d,\alpha)$ by direct computation, it suffices to bound $\|x\|_{2\alpha}^{2\alpha}$ by $M(d,\alpha)$ in each of the regimes for $\alpha$. We first resolve the trivial regimes of $\alpha\in (0,1/2]$ in Lemma~\ref{lemma:alpha<=1/2} and $d\leq 2$ in Lemma~\ref{lemma:d<=2}. 
\begin{lemma}\label{lemma:alpha<=1/2}
    For $\alpha \in (0,1/2]$ and $x\in \c F_d$, $\|x\|_{2\alpha}^{2\alpha}\geq 2^{1-\alpha}$.
\end{lemma}
\begin{proof}
Separating the coordinates by sign, WLOG, we can write 
$$x=(u_1,\dots,u_n,-v_1,\dots,v_m,0,\dots,0)$$
where $n+m\leq d$ and $u_i,v_j\geq 0$ for $i\in [n],j\in [m]$. For convenience, let $s\equiv S(u)=S(v)$. By Fact~\ref{fact:xp-subadd}, $f(x)=x^{2\alpha}$ is sub-additive. Then, 
   \begin{align*}
       \|x\|_{2\alpha}^{2\alpha}=\sum_{i\in [n]}u_i^{2\alpha}+\sum_{i\in [m]}v_i^{2\alpha}\geq S(u)^{2\alpha}+S(v)^{2\alpha}=2s^{2\alpha}
   \end{align*}
   From the $\|x\|_2=1$ constraint,
   \begin{align*}
       1&=\|x\|_2^2=\sum_{i\in [n]}u_i^2 +\sum_{i\in [m]}v_i^2\leq S(u)^2+S(v)^2=2s^2 
   \end{align*}
   Then, $ \|x\|_{2\alpha}^{2\alpha}\geq 2s^{2\alpha}\geq 2(1/2)^{\alpha}=2^{1-\alpha}$. 
\end{proof}


\begin{lemma}\label{lemma:d<=2}
    Theorem~\ref{thm:flat-main} holds for $d= 2$ and any $\alpha >0$.
\end{lemma}
\begin{proof}
   When $d=2$, $d^{-\alpha}((d-1)^\alpha+(d-1)^{1-\alpha})=2^{1-\alpha}$, meaning $M(2,\alpha)=2^{1-\alpha}$. Then, $S(x)=0$ gives $x_2=-x_1$, meaning $\|x\|_2^2=1$ implies $x_1=2^{-1/2}$.
   Thus, $\|x\|_{2\alpha}^{2\alpha}=2x_1^{2\alpha}=2^{1-\alpha}$.
\end{proof}

We are left to resolve Theorem~\ref{thm:flat-main} for $d\geq 3$, with $\alpha \in (1/2,1)$ or $\alpha> 1$. 

\subsection{Three-Value Reduction via the Equal Variables Method}
As before, we begin by handling various families of minimizers. 
\begin{lemma}\label{lemma:flat-lagrange}
   Define the following Lagrange function along with its respective gradient:
   \begin{equation}
       \begin{aligned}
           \c L (x;\lambda_1,\lambda_2)&=\|x\|_{2\alpha}^{2\alpha}-\lambda_1S(x)-\lambda_2(\|x\|_2^2-1)\\
           \p_j \c L(x;\lambda_1,\lambda_2)&= 2\alpha x_j |x_j|^{2(\alpha-1)}-2\lambda_2 x_j-\lambda_1
       \end{aligned}
   \end{equation}
   If $x\in \S^{d-1}$ with $S(x)=0$ is an extrema of $\c L(x;\lambda_1,\lambda_2)$, then there exists a choice of $\lambda_1,\lambda_2\in \R $ such that $\nabla_x \c L(x;\lambda_1,\lambda_2)=0$.
\end{lemma}
\begin{proof}
    This follows by direct computation and standard calculus.
\end{proof}

We first analyze minimizers with a zero coordinate via a similar argument to the proof of Lemma~\ref{lemma:obtuse-zero} for the obtuse case. 
\begin{lemma}\label{lemma:flat-zero}
    For $\alpha \in (1/2,1)\cup (1,\infty)$, if $x\in \c F_{d}$ has a zero coordinate and is the global extremizer of $\|x\|_{2\alpha}^{2\alpha}$, it must be of the form $(2^{-1/2},-2^{-1/2},0,\dots,0)$ up to permutation and negation. 
\end{lemma}
\begin{proof}
    Let $k\in [d]$ such that $x_k=0$. Then, from Lemma~\ref{lemma:flat-lagrange}, $\p_k \c L(x;\lambda_1,\lambda_2)=0$ gives $\lambda_1=0$. Then, 
    \begin{equation}
        0=\p_j \c L(x;\lambda_1,\lambda_2)=2x_j(\alpha |x_j|^{2(\alpha-1)}-\lambda_2)
    \end{equation}
    meaning $x_j\in\{0,\pm s\}$ for some $s>0$. Since $\|x\|_2^2=1$ and $S(x)=0$, there are the same number of positive and negative coordinates, say $r\geq 1$ of each. Then, $1=\|x\|_2^2=2rs^2$, meaning
    \begin{equation}
        \|x\|_{2\alpha}^{2\alpha}=2r s^{2\alpha}=(2r)^{1-\alpha}
    \end{equation}
    Then, $r=1$ minimizes this value for $\alpha \in (0,1)$ and maximizes it for $\alpha >1$. Since $1=2rs^2$, $s= 2^{-1/2}$, proving the claim.
\end{proof}

We now derive the following minimizer family.  

\begin{lemma}[Equation (22) of \citep{holevo2026conjecturetightnorminequality}]\label{lemma:3v-fam}
   For any $\alpha\in (1/2,1)$ or $\alpha > 1$, any extremizing $x$ in Theorem~\ref{thm:flat-main} has the form $x=(-s_0,\dots,-s_0,s_1,\dots,s_1,s_2,\dots,s_2)$, where $s_i$ appears $k_i$ times for $i\in \{0,1,2\}$, with $0\leq  s_1\leq s_2$ and $0\leq s_0$. 
\end{lemma}
\begin{proof}
Since the form in Lemma~\ref{lemma:flat-zero} lies within the claimed family, it suffices to prove the claim for $x$ with no zero coordinate.
From Lemma~\ref{lemma:flat-lagrange}, there exists a $\lambda_1,\lambda_2\in \R$ such that each $x_j$ is a root of the function
\begin{equation}
    f(t;\lambda_1,\lambda_2)\equiv 2\alpha t |t|^{2(\alpha-1)}-2\lambda_2 t-\lambda_1
\end{equation}
We now claim that $f(t)$ has at most three non-zero roots, which is equivalent to showing that $g(t)=\pm \lambda_1$ has at most three total solutions where $g(t)=2\alpha t^{2\alpha-1}-2\lambda_2 t$ for $t>0$. Taking a derivative,
\begin{equation}
    g'(t)=2\alpha (2\alpha-1)t^{2\alpha-2}-2\lambda_2
\end{equation}
If $\lambda_2\leq 0$, then $g'(t)>0$ for $t>0$, meaning $g(t)=|\lambda_1|$ and $g(t)=-|\lambda_1|$ each have at most one solution. Supposing $\lambda_2>0$, $g'(t)$ has at most one root for $t>0$. Thus, if $\alpha \in (1/2,1)$, $g$ increases from $0$ to a maximum and then tends to $-\infty$, in which case $g(t)=|\lambda_1|$ has at most two solutions while $g(t)=-|\lambda_1|$ has at most one. If $\alpha >1$, $g$ decreases from $0$ to a minimum and then tends to $+\infty$, in which case $g(t)=|\lambda_1|$ has at most one solution while $g(t)=-|\lambda_1|$ has at most two. 

Thus, there are at most three total solutions in each case, meaning $f(t;\lambda_1,\lambda_2)$ has at most three non-zero roots and thus $x$ has at most three distinct coordinate values. The claimed family follows by taking $-x$ instead of $x$ if $x$ has two distinct negative values.
\end{proof}

While Lemma~\ref{lemma:3v-fam} reduces our $d$-dimensional constrained optimization problem to simply constrained optimization over $(s_0,s_1,s_2,k_0,k_2)$, this solution family is still large. At the heart of our proof is a technique in symmetric inequalities known as the equal variables method \citep{cirtoaje2006algebraic}, which we use to eliminate the $k_2$ degree of freedom. 

\begin{theorem}[Equal Variables Theorem \citep{cirtoaje2006algebraic}]\label{thm:eqv}
For $d\geq 3$, let $a,x\in \R^d$ with $a_i\geq 0$ for all $i\in [d]$ and $0\leq x_1\leq \cdots\leq x_d$ such that
$$\sum_{j=1}^{d}x_j=\sum_{j=1}^{d}a_j,\quad \sum_{j=1}^{d}x_j^k =\sum_{j=1}^{d}a_j^k $$
where $a$ is fixed and $k\in \R$. If $k=0$, assume that $\prod_{i=1}^{d}x_i=\prod_{i=1}^{d}a_i>0$ instead. Let $f:(0,\infty)\to \R$ be a differentiable function such that $g:(0,\infty)\to \R$ defined by $g(x)=f'(x^{\frac 1 {k-1}})$ is strictly convex. Consider the following sum:
$$S_d\equiv \sum_{i=1}^{d}f(x_i)$$
Then, 
\begin{enumerate}
    \item If $k\leq 0$, $S_d$ is maximum for $x_1=\cdots=x_{d-1}\leq x_{d}$ and is minimum for $x_1\leq x_2=\cdots=x_d$.
    \item If $k>0$ and $\lim_{x\to 0^+}f(x)= -\infty$, $S_d$ is maximum for $x_1=\cdots=x_{d-1}\leq x_{d}$.
    \item  If $k>0$ and $\lim_{x\to 0^+}f(x)=\infty$, $S_d$ is minimum for $x_1=\cdots =x_{j-1}=0$ and $x_{j+1}=\cdots =x_d$, where $j\in [d]$. 
\end{enumerate}
If $f:[0,\infty)\to \R$ is continuous at $0$, the limit conditions in the last two statement are waived. 
\end{theorem}
\begin{proposition}\label{thm:eqv-unique}
   In each case of Theorem~\ref{thm:eqv}, any $x$ which attains the optimum value of $S_d$ must lie within the claimed family.
\end{proposition}
\begin{proof}
    This follows from the fact that the proof of the Equal Variables Theorem in Chapter 5 of \citep{cirtoaje2006algebraic} is proved by way of contradiction, showing that any $x$ not of the claimed form can be replaced with an $x$ of the form that achieves a strictly better value of $S_n$.
\end{proof}

Using Theorem~\ref{thm:eqv} with Proposition~\ref{thm:eqv-unique}, we are able to refine the three value family.

\begin{lemma}\label{lemma:k2=1}
For $\alpha \in (1/2,1)\cup(1,\infty)$, any global extremum $x$ can be expressed in the form of Lemma~\ref{lemma:3v-fam} with $k_2= 1$.
\end{lemma}
\begin{proof}
    From Lemma~\ref{lemma:flat-zero}, it suffices to consider $x$ with no zero coordinate. Then, from Lemma~\ref{lemma:3v-fam}, we have that 
    $$x=(-s_0,\dots,-s_0,s_1,\dots,s_1,s_2,\dots,s_2)$$
    where $0< s_1\leq s_2$, $s_0>0 $, and $s_i$ is repeated $k_i$ times, for $i\in \{0,1,2\}$. We seek to apply Theorem~\ref{thm:eqv} on the positive coordinates to show that $k_2=1$, and we begin by defining the positive block and negative block to be the following vectors:
    \begin{align*}
    x^+&=(s_1,\dots,s_1,s_2,\dots,s_2)\in \R^{d-k_0}\\
    x^-&=(s_0,\dots,s_0)\in \R^{k_0}
    \end{align*}
    Since $S(x)=0$ and no coordinates are zero, $k_0\geq 1$, so fix the negative block by choosing $s_0\in (0,1]$ and $k_0\in [d] $. 
    Then, the zero sum and $L^2$ constraint give the following feasible set over the positive block.
    \begin{align*}
        {\c S}(s_0,k_0)&\equiv \left\{x^+\in \R^{d-k_0}:x^+_i> 0,\quad 
        \sum_{i=1}^{d-k_0}x^+_i=k_0s_0,\quad \sum_{i=1}^{d-k_0}(x^+_i)^2=1-k_0s_0^2\right\}
    \end{align*}
    Then,  
    \begin{align*}
        \inf_{x\in {\c S}}\|x\|_{2\alpha}^{2\alpha}&=\inf_{\substack{
s_0\in(0,1],\,
k_0\in [d]\\
{\c S}(s_0,k_0)\neq\emptyset
}}\left(k_0s_0^{2\alpha}+\inf_{x^+\in {\c S}(s_0,k_0)}\|x^+\|^{2\alpha}_{2\alpha}\right)\quad (1/2<\alpha<1)\\ 
        \sup_{x\in {\c S}}\|x\|_{2\alpha}^{2\alpha}&=\sup_{\substack{
s_0\in(0,1],\,
k_0\in [d]\\
{\c S}(s_0,k_0)\neq\emptyset
}}\left(k_0s_0^{2\alpha}+\sup_{x^+\in {\c S}(s_0,k_0)}\|x^+\|^{2\alpha}_{2\alpha}\right)\quad (\alpha > 1)
    \end{align*}

    First, consider $\alpha \in (1/2,1)$ and the infimum over the positive block for any choice of $s_0,k_0$ such that ${\c S}(s_0,k_0)\neq \emptyset$. Suppose $d-k_0\geq 3$. We seek to apply Theorem~\ref{thm:eqv} for $f(x)=-x^{2\alpha}$, $k=2$, and any $a\in {\c S}(s_0,k_0)$ in order to maximize $S_{d-k_0}=-\|x^+\|_{2\alpha}^{2\alpha}$. We verify that $g''(x)=f^{(3)}(x)=-2\alpha(2\alpha-1)(2\alpha-2)x^{2\alpha-3}>0$ for $\alpha \in (1/2,1)$, meaning $g$ is strictly convex. $f$ is clearly continuous at $0$ and differentiable on $(0,\infty)$. The conditions on $a$ are satisfied by construction of the feasible set ${\c S}(s_0,k_0)$, and the entries of $x^+$ are positive and non-decreasing since $s_2\geq s_1$. Then, by the second statement of Theorem~\ref{thm:eqv}, the maximum of $S_{d-k_0}$, or equivalently $\inf_{x^+\in {\c S}(s_0,k_0)}\|x^+\|_{2\alpha}^{2\alpha}$, is attained only by  $x^+$ of the form:
    $$x^+=(x^+_1,\dots,x_1^+,x_{d-k_0}^+)$$
    where $0< x_{1}^+\leq x_{d-k_0}^+$. If $d-k_0=2$, by permutation, $x^+=(x_1^+,x_{2}^+)$ where $x_{2}^+\geq x_1^+> 0$, and if $d-k_0=1$, $x^+=(x_{d-k_0}^+)$. Thus, regardless of $d-k_0$, $x^+_{i}=x_j^+$ for any $i,j\in [1,d-k_0-1]$, and $x^{+}_i\leq x^+_{d-k_0}$ for $i \in [1,d-k_0]$. 
    
    Next, consider $\alpha > 1$ and the supremum over the positive block for any choice of $s_0,k_0$ such that ${\c S}(s_0,k_0)\neq \emptyset$. Suppose $d-k_0\geq 3$. We again seek to apply Theorem~\ref{thm:eqv}, this time for $f(x)=x^{2\alpha}$, $k=2$, and any $a\in {\c S}(s_0,k_0)$ in order to maximize $S_{d-k_0}=\|x^+\|_{2\alpha}^{2\alpha}$. We verify that $g''(x)=2\alpha(2\alpha-1)(2\alpha-2)x^{2\alpha-3}>0$ for $\alpha > 1$, meaning $g$ is strictly convex, and the rest of the conditions follow in the same manner as the $\alpha \in (1/2,1)$ case. Then, by the second statement of Theorem~\ref{thm:eqv}, the maximum of $S_d$, or equivalently $\sup_{x^+\in {\c S}(s_0,k_0)}\|x^+\|_{2\alpha}^{2\alpha}$, is attained only by $x^+$ of the form:
    $$x^+=(x^+_1,\dots,x_1^+,x_{d-k_0}^+)$$
    where $0< x_{1}^+\leq x_{d-k_0}^+$. Note that this is the same form as the $\alpha \in (1/2,1)$ case, and by the same argument as before, the $d-k_0\in \{1,2\}$ cases are subsumed by this. 
    
    Then, identifying $s_1\equiv x_1^+$ and $s_2\equiv x_{d-k_0}^+$, we see that $k_2\leq 1$. Since this form holds on the positive block for any valid choice of $(s_0,k_0)$, the claim follows, where we use the fact that $s_1=s_2$ is permitted to compress the $k_2=0$ case.
\end{proof}
\subsection{Proof of Theorem~\ref{thm:flat-main}}
Equipped with the much simpler family in Lemma~\ref{lemma:k2=1}, Lemmas~\ref{lemma:a1/2_k0>=2}, \ref{lemma:alpha1/2_k0=1}, \ref{lemma:alpha>1_k2=1_k0>=2}, and \ref{lemma:alpha>1_k2=1_k0=1} reduce the non-trivial regimes of Theorem~\ref{thm:flat-main} to the three one-dimensional optimization problems of Lemmas~\ref{lemma:m2-left-min}, \ref{lemma:m2-right-min}, and \ref{lemma:m1-min}. 

\begin{lemma}\label{lemma:a1/2_k0>=2}
    For $\alpha \in (1/2,1)$ and $d\geq 4$, any global minimizer with $k_2=1$ and $k_0\geq 2$ satisfies
    $$\|x\|_{2\alpha}^{2\alpha}\geq \min\{2^{1-\alpha},d^{-\alpha}((d-1)^{\alpha}+(d-1)^{1-\alpha})\}$$
\end{lemma}
\begin{proof}
    Since the form in Lemma~\ref{lemma:flat-zero} achieves $\|x\|_{2\alpha}^{2\alpha}=2^{1-\alpha}$, it suffices to consider $x$ with no zero coordinate. From Lemma~\ref{lemma:k2=1}, 
    \begin{align*}
        \|x\|_{2\alpha}^{2\alpha}\geq \inf \{k_0s_0^{2\alpha}+(d-k_0-1)s_1^{2\alpha}+s_2^{2\alpha}:k_0s_0=(d-k_0-1)s_1+s_2,\, k_0s_0^2+(d-k_0-1)s_1^2+s_2^2=1\}
    \end{align*}
    where $0< s_1\leq s_2$, $0< s_0$, and $2\leq k_0\leq d-1$. By Lemma~\ref{fact:lp-monotonicity} for $p=2\alpha<2=q$, 
    \begin{align*}
        \left((d-k_0-1)s_1^{2\alpha}+s_2^{2\alpha}\right)^{1/2\alpha}\geq \left((d-k_0-1)s_1^2+s_2^2\right)^{1/2}=(1-k_0s_0^2)^{1/2}
    \end{align*}
    Then, we can lower bound the objective as
    $$k_0s_0^{2\alpha}+(d-k_0-1)s_1^{2\alpha}+s_2^{2\alpha}\geq k_0s_0^{2\alpha}+(1-k_0s_0^2)^{\alpha}$$
    This motivates the substitution $u=k_0s_0^2> 0$ such that the right hand side is $k_0^{1-\alpha}u^{\alpha}+(1-u)^{\alpha}$. We now determine the feasible set of $u$. From the two constraints,
    \begin{align*}
        1-u=1-k_0s_0^2=(d-k_0-1)s_1^2+s_2^2\leq ((d-k_0-1)s_1+s_2)^2=(k_0s_0)^2=k_0u\Rightarrow u\geq \frac{1}{k_0+1}
    \end{align*}
    Further, by Cauchy-Schwarz,
    \begin{align*}
        k_0u&=((d-k_0-1)s_1+s_2)^2\leq (d-k_0)((d-k_0-1)s_1^2+s_2^2)=(d-k_0)(1-u)\Rightarrow u\leq \frac{d-k_0}{d}
    \end{align*}
    Thus, 
    \begin{align*}
         \|x\|_{2\alpha}^{2\alpha}&\geq \inf \{k_0s_0^{2\alpha}+(d-k_0-1)s_1^{2\alpha}+s_2^{2\alpha}:k_0s_0=(d-k_0-1)s_1+s_2,\, k_0s_0^2+(d-k_0-1)s_1^2+s_2^2=1\}\\
         &\geq \inf \left\{f(u;k_0):\frac{1}{k_0+1}\leq u\leq \frac{d-k_0}{d},\,2\leq k_0\leq d-1\right\}
    \end{align*}
    where $f(u;k_0)=k_0^{1-\alpha}u^\alpha+(1-u)^{\alpha}$. Now, 
    \begin{align*}
        f''(u;k_0)&= k_0^{1-\alpha}\alpha(\alpha-1)u^{\alpha-2}+\alpha(\alpha-1)(1-u)^{\alpha-2}<0
    \end{align*}
    on $u\in (0,1)$, since $\alpha < 1$. Thus, $f(u;k_0)$ is strictly concave on $(0,1)$, meaning
    \begin{align*}   
    \|x\|_{2\alpha}^{2\alpha}&\geq \inf \left\{f(u;k_0):\frac{1}{k_0+1}\leq u\leq \frac{d-k_0}{d},\,2\leq k_0\leq d-1\right\}\\&=\min_{2\leq k_0\leq d-1}\min \left \{f\left(\frac{1}{k_0+1};k_0\right),f\left(\frac{d-k_0}{d};k_0\right)\right\}
    \end{align*}
    We now directly calculate the minimum of each branch. For $t\in [2,d-1]$, 
    \begin{align*}
        f_l(t)&\equiv f\left(\frac{1}{t+1};t\right)=\frac{t^{\alpha}+t^{1-\alpha}}{(t+1)^{\alpha}}\\
        f_r(t)&\equiv f\left(\frac{d-t}{d};t\right)=d^{-\alpha}(t^{1-\alpha}(d-t)^\alpha+t^\alpha)
    \end{align*}
    Since $d\geq 3$ and $k_0\geq 2$, by Lemmas~\ref{lemma:m2-left-min} and \ref{lemma:m2-right-min}, we are done.
\end{proof}

\begin{lemma}\label{lemma:alpha1/2_k0=1}
    For $\alpha \in (1/2,1)$ and $d\geq 3$, any global minimizer with $k_2=1$ and $k_0=1$ satisfies
    $$\|x\|_{2\alpha}^{2\alpha}\geq \min\{2^{1-\alpha},d^{-\alpha}((d-1)^{\alpha}+(d-1)^{1-\alpha})\}$$
\end{lemma}
\begin{proof}
    As in Lemma~\ref{lemma:a1/2_k0>=2}, it suffices to consider $x$ with no zero coordinate, meaning we are considering the form $(-s_0,s_1,\dots,s_1,s_2)$ such that $0< s_1\leq s_2$ and $s_0 > 0$. Writing $t\equiv s_2/s_1\geq 1$, the constraints become:
    \begin{align*}
       0&=-s_0+(d-2)s_1+s_2\Rightarrow s_0=(d-2+t)s_1\\
       1&=s_0^2 +(d-2)s_1^2+s_2^2=(d-2+t)^2s_1^2+(d-2)s_1^2+t^2s_1^2\Rightarrow s_1^2=\frac{1}{(d-2+t)^2 +d-2+t^2 }
    \end{align*}
    Then, the objective becomes
    \begin{align*}
        s_0^{2\alpha}+(d-2)s_1^{2\alpha}+s_2^{2\alpha}&=((d-2+t)^{2\alpha}+(d-2)+t^{2\alpha})s_1^{2\alpha}=\frac{(d-2+t)^{2\alpha}+(d-2)+t^{2\alpha}}{((d-2+t)^{2} +d-2+t^2 )^\alpha}\equiv f(t)
    \end{align*}
    Thus, $\|x\|_{2\alpha}^{2\alpha}\geq \min_{t\geq 1}f(t)$. By Lemma~\ref{lemma:m1-min}, this minimum is exactly the claim lower bound. 
\end{proof}

\begin{lemma}\label{lemma:alpha>1_k2=1_k0>=2}
For $\alpha>1$ and $d\geq 4$, any global maximizer with $k_2=1$ and $k_0\geq 2$ satisfies 
$$\|x\|_{2\alpha}^{2\alpha}\leq \max\{2^{1-\alpha},d^{-\alpha}((d-1)^{\alpha}+(d-1)^{1-\alpha})\}$$
\end{lemma}
\begin{proof}
We follow an analogous argument to Lemma~\ref{lemma:a1/2_k0>=2} and consider $x$ with no zero coordinate. By Lemma~\ref{fact:lp-monotonicity} for $p=2<2\alpha=q$, 
\begin{align*}
    ((d-k_0-1)s_1^{2\alpha}+s_2^{2\alpha})^{1/(2\alpha)}\leq ((d-k_0-1)s_1^2+s_2^2)^{1/2}=(1-k_0s_0^2)^{1/2}
\end{align*}
Then, by the same substitution $u=k_0s_0^2$ and an analogous argument, 
\begin{align*}
    \|x\|_{2\alpha}^{2\alpha}\leq \max_{2\leq k_0\leq d-1}\max \left\{f\left(\frac{1}{k_0+1};k_0\right),f\left(\frac{d-k_0}{d};k_0\right)\right\}
\end{align*}
where $f(u;k_0)=k_0^{1-\alpha}u^\alpha+(1-u)^\alpha$ as before. Then, by Lemma~\ref{lemma:m2-left-min} and Lemma~\ref{lemma:m2-right-min}, which is valid for $d\geq 4$ and $k_0\geq 2$, we have that these maximums give exactly the claimed upper bound. 
\end{proof}
\begin{lemma}\label{lemma:alpha>1_k2=1_k0=1}
   For $\alpha > 1$ and $d\geq 3$, any global maximizer with $k_2=1$ and $k_0=1$ satisfies 
$$\|x\|_{2\alpha}^{2\alpha}\leq \max\{2^{1-\alpha},d^{-\alpha}((d-1)^{\alpha}+(d-1)^{1-\alpha})\}$$
\end{lemma}
\begin{proof}
By an analogous argument to Lemma~\ref{lemma:alpha1/2_k0=1}, we can consider $x$ with no zero coordinate and substitute $t=s_2/s_1\geq 1$, yielding:
$$\|x\|_{2\alpha}^{2\alpha}\leq \max_{t\geq 1}\frac{(d-2+t)^{2\alpha}+(d-2)+t^{2\alpha}}{((d-2+t)^2+d-2+t^2)^\alpha}$$
By Lemma~\ref{lemma:m1-min}, this maximum is exactly the claimed upper bound. 
\end{proof}

\begin{lemma}
Theorem~\ref{thm:flat-main} holds for any $\alpha \in (1/2,1)\cup (1,\infty)$ and $d\geq 3$. 
\end{lemma}
\begin{proof}
    For any $\alpha \in (1/2,1)\cup(1,\infty)$, by Lemma~\ref{lemma:k2=1}, any global extremum with no zero coordinate has the form:
    $$x=(-s_0,\dots,-s_0,s_1,\cdots,s_1,s_2)$$
    where $s_0\geq 0$ and $0\leq s_1\leq s_2$, where $s_i$ occurs $k_i$ times for $i\in \{0,1\}$. 
    
    For $d\geq 4$, by Lemma~\ref{lemma:alpha>1_k2=1_k0>=2} and \ref{lemma:alpha>1_k2=1_k0=1}, we are done for $\alpha > 1$, and by Lemmas~\ref{lemma:a1/2_k0>=2} and \ref{lemma:alpha1/2_k0=1}, we are done for $\alpha \in (1/2, 1)$ as well. Now, for $d=3$, since $k_2=1$, $k_0\in \{1,2\}$. If $k_0=1$, then Lemmas~\ref{lemma:alpha1/2_k0=1} and \ref{lemma:alpha>1_k2=1_k0=1} together finish the lower and upper bounds. If $k_0=2$, then $x=(-s_0,-s_0,s_2)$, where $s_2=2s_0$ from the zero sum constraint. Then, 
    \begin{align*}
        1&=\|x\|^2_2=6s_0^2 \\
        \|x\|_{2\alpha}^{2\alpha}&=(2+2^{2\alpha})s_0^{2\alpha}=(2+2^{2\alpha})6^{-\alpha}=3^{-\alpha}((3-1)^\alpha+(3-1)^{1-\alpha})
    \end{align*}
    which satisfies the claimed bound as well as the equality case for $d=3$ and $k_0=2$. 

    We now consider the equality case more fully. Since Lemma~\ref{lemma:flat-zero} handles any global extremum with a zero coordinate, consider a global extremum, $x$, with no zero coordinate that saturates the bound. For $d\geq 4$ and $k_0\geq 2$, the $\ell^p $-monotonicity from Lemma~\ref{fact:lp-monotonicity} used on the positive block in Lemmas~\ref{lemma:a1/2_k0>=2} and \ref{lemma:alpha>1_k2=1_k0>=2} must then be an equality:
    \begin{equation}
        (d-k_0-1)s_1^{2\alpha}+s_2^{2\alpha}= ((d-k_0-1)s_1^2+s_2^2)^\alpha
    \end{equation}
    which only occurs if there is at most one non-zero value in the positive block. Since $x$ has no zero coordinate, we must have the form $x=(-s_0,\dots,-s_0,s_2)$. Since $\|x\|_2=1$ and $S(x)=0$ imply $s_2=(d-1)s_0$ and $d(d-1)s_0^2=1$, the claimed non-zero value equality family in Theorem~\ref{thm:flat-main} follows. For $d\geq 3$ and $k_0=1$, Lemma~\ref{lemma:m1-min} used in both Lemmas~\ref{lemma:alpha1/2_k0=1} and \ref{lemma:alpha>1_k2=1_k0=1} states that equality occurs only for $t=1$ or $t\to\infty$ where $t=s_2/s_1$, unless $(d,\alpha)=(3,2)$ in which case all $t\geq 1$ saturate the bound. Discarding $t\to \infty$ as it corresponds to $s_1\to 0$ though $x$ has no zero coordinate, $t=1$ gives $x=(-s_0,\dots,-s_0,s_2)$ and thus the non-zero desired equality family. The exception $(d,\alpha)=(3,2)$ is derivative of the fact that every $x\in \R^3 $ with $S(x)=0$ and $\|x\|_2=1$ has $\|x\|_4^4=1/2$. We can see this with the following manipulation:
    \begin{equation}
        \begin{aligned}
            x_1x_2+x_2x_3+x_1x_3&=\frac{S(x)^2-\|x\|_2^2}{2}=-1/2 \\
            x_1^2x_2^2+x_2^2x_3^2+x_1^2x_3^2&=(x_1x_2+x_2x_3+x_1x_3)^2-2x_1x_2x_3S(x)=1/4\\
            \|x\|_4^4 &=\|x\|_2^4 -2(x_1^2x_2^2+x_2^2x_3^2+x_1^2x_3^2)=1/2
        \end{aligned}
    \end{equation}
    This proves every case of Theorem~\ref{thm:flat-main}.
\end{proof}

\section{Discussion}
By proving Theorems~\ref{thm:obtuse-59-67} and \ref{thm:flat-main}, we confirm the globally information-optimal measurement for obtuse and flat quantum pyramids, resolving the conjecture of \citep{englert2010well}. We take this as a further demonstration of how powerful the dual program of \citep{holevo2025quantum} can be, and we hope that leveraging this technique for other pure state ensembles will yield both proofs of optimality as well as non-trivial entropy inequalities. An enticing example is the pretty-good measurement \citep{peres1991optimal} for double trines \citep{decker2009symmetric}, but more importantly, analyzing the minimizers of the resulting entropy inequalities may provide insight into the information-optimal measurement itself as there are many families for which this is not known \citep{dall2014accessible,dall2014tight}. A success of this form for a symmetric mixed state ensemble would be spectacular.

 Regarding the tight $\ell^p$ conjecture for zero-sum vectors in Theorem~\ref{thm:flat-main}, we would like to acknowledge the independent proof of \citep{zhang2026proof}. While we also tackle the conjecture by analyzing minimizers and in fact arrive at the same family in Lemma~\ref{lemma:k2=1} as Lemma 10 of \citep{zhang2026proof}, we do so via the equal variables method which considerably expedites the derivation. Our subsequent parameterization and proof of the resulting scalar inequality is also different. In fact, we posit that the equal variables method may enable the study of the following family of optimization problems:
 \begin{equation}
     M_p(d,\alpha)=\begin{cases}
         \min_{x\in {\c F}^p_d}\|x\|_{p\alpha}^{p\alpha} & \alpha < 1\\
         \max_{x\in {\c F}^p_d}\|x\|_{p\alpha}^{p\alpha} & \alpha > 1
     \end{cases}
 \end{equation}
 where ${\c F}_d^p=\{x\in \R^d:\|x\|_p=1\text{ and }\sum_{j=1}^d x_j=0\}$ for $p>0$, with $p=2$ being Theorem~\ref{thm:flat-main}. We remark that $M_p(d,\alpha)$ resembles quantities that arise in the study of information geometry and escort distributions \citep{ohara2010dually} as well as non-linear $p$-log-Sobolev inequalities on the complete graph \citep{polyanskiy2019improved, gu2023non}, or particularly the modified log-Sobolev inequality for $p=1$ \citep{bobkov2006modified}. 
 
 This is all to say that while the quantum pyramids conjecture is resolved, the legacy of the quantum pyramid is far from over.

\section{Acknowledgements}
We would like to thank Prof. Alexander Holevo and Prof. Andrey Utkin for verifying our proof of Theorem~\ref{thm:flat-main} as well as Prof. Sitan Chen for being a wonderful advisor. ChatGPT Pro 5.4 contributed calculations used to prove the technical lemmas of the appendix and also found the equal variables method in the literature, all of which was checked by hand.

\bibliographystyle{alpha}
\bibliography{refs}

\clearpage
\appendix

\section{Technical Details}
\subsection{Key Algebraic Reciprocal Inequality for Section~\ref{sec:obtuse}}\label{x:obtuse}
We now prove the final algebraic inequality which jettisons minimizers from the three distinct value family for the entropy inequality in Theorem~\ref{thm:obtuse-59-67}.
\begin{lemma}\label{lemma:obtuse-reciprocal}
    Let $0<\alpha<x<1<\beta$ satisfy $\alpha e^\alpha=xe^{-x}=\beta e^{-\beta}$. Then, 
    \begin{equation}
        \frac{1}{x(1-x)}<\frac{1}{\alpha(1+\alpha)}+\frac{1}{\beta(\beta-1)}
    \end{equation}
\end{lemma}
\begin{proof}
    We proceed by first upper bounding $\alpha$ and $\beta$ in terms of $x$ and then proving the resulting single variable inequality. 

    We begin with $\beta$, observing that $xe^{-x}=\beta e^{-\beta}$ is equivalent to saying that the logarithmic mean of $x$ and $\beta$ is $1$. Recall that $L^3\geq G^2A$ where $L, G, A$ denote the logarithmic, geometric, and arithmetic mean, with equality attained only when the two inputs are equal \citep{leach1983extended}. Then, by the quadratic formula, 
    \begin{equation}\label{lemma:obtuse-alg-beta-ub}
        x\beta \frac{x+\beta}{2}<1\Rightarrow x\beta^2 +x^2\beta -2<0 \Rightarrow \beta < \frac{-x+\sqrt{x^2+8/x}}{2}\equiv B(x)
    \end{equation}
    We now continue to $\alpha$, and claim that:
   \begin{equation}\label{lemma:obtuse-alg-alpha-ub}
       \alpha < \frac{x}{1+2x+x^3/2}\equiv A(x)
   \end{equation} 
    Since $t\mapsto te^t $ is strictly increasing, it suffices to show that $A(x)e^{A(x)}>xe^{-x}$, or equivalently:
    \begin{equation}
        G(x)\equiv x+\frac{x}{1+2x+x^3/2}-\ln (1+2x+x^3/2)>0
    \end{equation}
    Differentiating, 
    \begin{equation}
        G'(x)=1+\frac{1+2x+x^3/2-x(2+3x^2/2)}{(1+2x+x^3/2)^2}-\frac{2+3x^2/2}{1+2x+x^3/2}=\frac{x^2(x-1)(x^3-2x^2+6x-10)}{(x^3+4x+2)^2}
    \end{equation}
    where one can verify the last equality by direct expansion. Now, consider the polynomial $p(x)\equiv x^3-2x^2+6x-10$, whose derivative is $p'(x)=3x^2-4x+6>3-4+6>0$ on $(0,1)$. Since $p(1)=-5$, $p(x)$ is negative on $(0,1)$, meaning $G'(x)>0$ on $(0,1)$. Since $G(0)=0$ as well, $G(x)>0$ on $(0,1)$, proving \eqref{lemma:obtuse-alg-alpha-ub}. Applying the upper bounds in \eqref{lemma:obtuse-alg-beta-ub} and \eqref{lemma:obtuse-alg-alpha-ub}, one can verify by direct expansion that:
    \begin{equation}\label{eq:obtuse-alg-ratx}
    \begin{aligned}
        \frac{1}{\alpha(1+\alpha)}+\frac{1}{\beta(\beta-1)}-\frac{1}{x(1-x)}&>\frac{1}{A(x)(1+A(x))}+\frac{1}{B(x)(B(x)-1)}-\frac{1}{x(1-x)}\\
        &=\frac{x^2(Q(x)-x^{3/2}\sqrt{x^3+8}\,P(x))}{x(x-1)(x^{3/2}-\sqrt{x^3+8})(x^{3/2}+2\sqrt x-\sqrt{x^3+8})(x^3+6x+2)}
    \end{aligned}
    \end{equation}
    where 
    \begin{equation}\label{eq:obtuse-alg-PQ}
        \begin{aligned}
        P(x)&\equiv x^5+7x^3+4x^2+10x+14\\
Q(x)&\equiv x^8+7x^6+8x^5+10x^4+42x^3+8x^2+40x-8
        \end{aligned}
    \end{equation}
    Analyzing the denominator of \eqref{eq:obtuse-alg-ratx}, the first and last terms are positive on $(0,1)$ while the second and third terms are negative. We also claim that the fourth is negative: 
    \begin{equation}
        x^{3/2}+2\sqrt x<\sqrt{x^3+8}\Leftrightarrow x^{3}+4x^2+4x<x^3+8\Leftrightarrow x(x+1)<2
    \end{equation}
    which is clearly true on $(0,1)$. Thus, the denominator of \eqref{eq:obtuse-alg-ratx} is negative, meaning it suffices to show that the numerator is negative, or equivalently:
    \begin{equation}\label{eq:obtuse-alg-PQ-ineq}
        Q(x)<x^{3/2}\sqrt{x^3+8}\, P(x)
    \end{equation}
    Now, since $P(x)>0$ from \eqref{eq:obtuse-alg-PQ}, if $Q(x)<0$, we are done since the numerator is immediately negative. Thus, suppose $Q(x)\geq 0$, meaning proving \eqref{eq:obtuse-alg-PQ-ineq} is equivalent to proving its square:
    \begin{equation}\label{eq:obtuse-alg-PQ-ineq2}
        0<x^3(x^3+8)P(x)^2-Q(x)^2=16(1-x)^2T(x)
    \end{equation}
    where $T(x)=x^6+2x^5+20x^3-24x^2+32x-4$, which can be verified by direct expansion. 
    From \eqref{eq:obtuse-alg-PQ}, $Q'(x)>0$ on $(0,1)$ and $Q(1/7)=-11502504/5764801<0$, meaning $x>1/7$. Then, 
    \begin{equation}
        T'(x)=6x^5+10x^4+60x^2-48x+32=6x^5+10x^4+60(x-2/5)^2+112/5>0
    \end{equation}
    meaning $T(x)>0$ on $(0,1)$, since $T(1/7)=16479/117649>0$ as well. This proves \eqref{eq:obtuse-alg-PQ-ineq2}, and thus the lemma.
\end{proof}
We can equivalently view the condition $\alpha e^\alpha=xe^{-x}=\beta e^{-\beta}$ for $0<\alpha<x<1<\beta$ as setting $\alpha=W_0(xe^{-x})$ and $\beta=-W_{-1}(-xe^{-x})$, where $W_0$ and $W_{-1}$ correspond to the principal real and lower real branch of the Lambert-$W$ function.

\subsection{Scalar Inequalities for Section~\ref{sec:flat}}\label{x:flat}
\begin{fact}\label{fact:xp-subadd}
    For $p\in (0,1]$, $f(x)=x^p$ is subadditive. 
\end{fact}
\begin{proof}
    We desire to show that for $x,y\geq 0$, $(x+y)^p\leq x^p+y^p$. Taking $t=y/x\geq 0$, WLOG, it suffices to prove that $(1+t)^p\leq 1+t^p$. Let $f(t)=1+t^p-(1+t)^p$. Then,
    \begin{align*}
       f'(t)=p(t^{p-1} -(1+t)^{p-1})\geq 0
    \end{align*}
    since $t\mapsto t^{p-1}$ is a decreasing function for $p\in (0,1]$. Since $f(0)=0$, $f(t)\geq 0$ for all $t\geq 0$, and we are done.
\end{proof}
\begin{fact}\label{fact:lp-monotonicity}
   For $x \in \R^n $ and $0<p<q\leq \infty$, $\|x\|_q\leq \|x\|_p\leq n^{1/p-1/q}\|x\|_q$. 
\end{fact}
\begin{proof}
    By H\"older's inequality, since $q/p>1$ and $q/(q-p)>1$ are valid conjugate exponents,
    \begin{align*}
        \|x\|_p^p=\sum_{i=1}^{n}|x_i|^p\cdot 1\leq \left(\sum_{i=1}^{n}|x_i|^{p\cdot q/p}\right)^{p/q}\left(\sum_{i=1}^{n}1^{q/(q-p)}\right)^{(q-p)/q}=\|x\|_q^{p}n^{1-p/q}
    \end{align*}
    which proves the upper bound after taking $(\cdot)^{1/p}$ on both sides. For the lower bound, let $y=x/\|x\|_p$ such that $\|y\|_p=1$, meaning $|y_i|\leq 1$ for each $i\in [n]$. Then, since $p<q$, 
    \begin{align*}
        \|y\|_q^q&=\sum_{i=1}^{n}|y_i|^q\leq \sum_{i=1}^{n}|y_i|^p=\|y\|_p^p=1
    \end{align*}
    Thus, $\|x\|_q=\|x\|_p\|y\|_q\leq \|x\|_p$, proving the lower bound.
\end{proof}
\begin{fact}
    For $d\geq 2 $, if $\alpha \in (1/2,1)$, 
    $$\min\{2^{1-\alpha},d^{-\alpha}((d-1)^\alpha+(d-1)^{1-\alpha})\}=\begin{cases}
        2^{1-\alpha} & d\leq d(\alpha)\\
        d^{-\alpha}((d-1)^\alpha+(d-1)^{1-\alpha}) & d > d(\alpha)
    \end{cases}$$
    where $d(\alpha)$ is the largest root of $D(d)=d^{-\alpha}((d-1)^\alpha+(d-1)^{1-\alpha})-2^{1-\alpha}$. If $\alpha > 1$, 
    $$\max\{2^{1-\alpha},d^{-\alpha}((d-1)^\alpha+(d-1)^{1-\alpha})\}=\begin{cases}
        2^{1-\alpha} & d\leq d(\alpha)\\
        d^{-\alpha}((d-1)^\alpha+(d-1)^{1-\alpha}) & d > d(\alpha)
    \end{cases}$$
    
\end{fact}
\begin{proof}
    Let $L(\alpha)=2^{1-\alpha}$ and $R(\alpha,d)=d^{-\alpha}((d-1)^\alpha+(d-1)^{1-\alpha})$. Consider $d=2$. 
    \begin{align*}
        R(\alpha,2)&=2^{-\alpha}((2-1)^\alpha+(2-1)^{1-\alpha})=2^{1-\alpha}=L(\alpha)
    \end{align*}
    meaning $d=2$ is a root of $D(d)$. We proceed by substituting $t\equiv (d-1)^{-1}\in (0,1]$. Then, 
    \begin{align*}
        R(\alpha,d)&=(1+t^{-1})^{-\alpha}(t^{-\alpha}+t^{\alpha-1})=(1+t)^{-\alpha}(1+t^{2\alpha-1})\equiv f(t)
    \end{align*}
    Differentiating,
    \begin{align*}
        f'(t)&=-\alpha(1+t)^{-\alpha-1}(1+t^{2\alpha-1})+(2\alpha-1)t^{2\alpha-2}(1+t)^{-\alpha}
        =(1+t)^{-\alpha-1}g(t)
    \end{align*}
    where
    \begin{align*}
        g(t)&\equiv (\alpha-1)t^{2\alpha-1}+(2\alpha-1)t^{2\alpha-2}-\alpha\\
        g'(t)&=(\alpha-1)(2\alpha-1)t^{2\alpha-2}+(2\alpha-1)(2\alpha-2)t^{2\alpha-3}=(2\alpha-1)(\alpha-1)\underbrace{(t+2)t^{2\alpha-3}}_{>0}
    \end{align*}
    We now break into cases for $\alpha$. 

    Consider $\alpha \in (1/2,1)$. Then, $g'(t)<0$ on $(0,1]$. Further, $g(1)=2\alpha-2<0$ and
    \begin{align*}
        \lim_{t\to 0^+}g(t)&=\lim_{t\to0^+}\frac{2\alpha-1}{t^{2(1-\alpha)}}=\infty
    \end{align*}
    which all together mean that $g(t)$ has a unique zero at some $t^*\in (0,1)$, at which it changes from positive to negative. Since ${\rm sign}(f'(t))={\rm sign}(g(t))$, it follows that $f(t)$ has a unique maximum at $t=t^*$. Now, evaluating the boundary, for $\alpha \in (1/2,1)$, 
    \begin{align*}
        f(1)&=2^{1-\alpha}>1,\quad \lim_{t\to 0^+}f(t)=1
    \end{align*}
    Together with the fact that $f'(t)$ has a unique zero in $(0,1)$ which yields a maximum of $f(t)$, it follows that there exists a unique $t'\in (0,1)$ such that $f(t')=2^{1-\alpha}$, with $t'\leq t^*$. Further, $f'(t)>0$ on $(0,t^*)$ implies that $f(t)<f(t')=2^{1-\alpha}$ for $t<t'$. Finally, since $f(t')=f(1)=2^{1-\alpha}$ and $f'$ has a unique zero on $(t',1)$ that is a local maxima of $f(t)$, it follows that $f(t)\geq 2^{1-\alpha}$ for $t \in (t',1)$. Thus, for $t\in (0,1]$,
    $$\min\{2^{1-\alpha},f(t)\}=\begin{cases}
        f(t)& t < t'\\
        2^{1-\alpha} & t \geq  t'
    \end{cases}\Rightarrow \min\{2^{1-\alpha},R(\alpha,d)\}=\begin{cases}
        R(\alpha,d) & d > d(\alpha)\\
        2^{1-\alpha} & d \leq d(\alpha)
    \end{cases}$$
    where $d(\alpha)=1+1/t'$ is the unique root of $D(d)$ in $(2,\infty)$. Since $d=2$ is the only other relevant root, $d(\alpha)$ is the largest root. This shows the claim for $\alpha \in (1/2,1)$. 

    Next, consider $\alpha>1$. Then, $g'(t)>0$ on $(0,1]$, $g(1)=2\alpha-2>0$, and 
    \begin{align*}
        \lim_{t\to 0^+}g(t)&=-\alpha<0
    \end{align*}
    which all together mean that $g(t)$ has a unique zero at some $t^*\in (0,1)$, at which it changes from negative to positive. This corresponds to a unique minimum of $f(t)$ at $t=t^*$. Recall that $\lim_{t\to 0^+}f(t)=1$ and $f(1)=2^{1-\alpha}<1 $ for $\alpha > 1$. Together with the fact that $f'(t)$ has a unique zero on $(0,1)$ at which $f(t)$ is a local minimum, it follows that there exists a unique $t'\in (0,1)$ such that $f(t')=2^{1-\alpha}$, with $t'\leq t^*$. Since $f'(t)<0$ on $(0,t^*)$, we have that $f(t)>f(t')=2^{1-\alpha}$ for $t<t'$. Finally, since $f(t')=f(1)=2^{1-\alpha}$ and $f'(t)$ has a unique zero on $(t',1)$ that is a local minima of $f(t)$, it follows that $f(t)\leq 2^{1-\alpha}$ for $t\in (t',1)$. Thus, for $t\in (0,1]$,
     $$\max\{2^{1-\alpha},f(t)\}=\begin{cases}
        f(t)& t < t'\\
        2^{1-\alpha} & t \geq  t'
    \end{cases}\Rightarrow \max\{2^{1-\alpha},R(\alpha,d)\}=\begin{cases}
        R(\alpha,d) & d > d(\alpha)\\
        2^{1-\alpha} & d \leq d(\alpha)
    \end{cases}$$
    which proves the claim for $\alpha > 1$. Thus, we are done.
\end{proof}

\begin{lemma}[Rolle's Theorem]\label{lemma:rolle}
   If $f:[a,b]\to \R$ is continuous on $[a,b]$ and differentiable on $(a,b)$ with $f(a)=f(b)$, then there exists at least one $c\in (a,b)$ with $f'(c)=0$. This also implies that if $f$ has at least $n$ distinct roots $r_1< \cdots<r_n$, then $f'$ has at least $n-1$ distinct roots $r_1'<\cdots<r_{n-1}'$, where $r_i'\in (r_i,r_{i+1})$ for $i\in [n-1]$. 
\end{lemma}
\begin{proof}
   The first statement is a standard fact from calculus. The second statement follows by applying Rolle's theorem on each interval $[r_i,r_{i+1}]$ for each $i\in [n-1]$. 
\end{proof}


\begin{lemma}\label{lemma:m2-right-min}
    Let $f(t)=d^{-\alpha}(t^{1-\alpha}(d-t)^{\alpha}+t^\alpha)$ for $d\geq 4$. Then, for $\alpha \in (1/2,1)$,
    $$\min_{t\in [2,d-1]}f(t)\geq \min\{2^{1-\alpha},d^{-\alpha}((d-1)^\alpha+(d-1)^{1-\alpha})\}$$
    while for $\alpha > 1$, 
    $$\max_{t\in [2,d-1]}f(t)\leq \max \{2^{1-\alpha},d^{-\alpha}((d-1)^\alpha+(d-1)^{1-\alpha})\}$$
\end{lemma}
\begin{proof}
Taking derivatives,
\begin{align*}
    f'(t)&=d^{-\alpha}((1-\alpha)t^{-\alpha}(d-t)^\alpha -\alpha t^{1-\alpha}(d-t)^{\alpha-1} +\alpha t^{\alpha-1} )\\
    f''(t)&=d^{-\alpha}((1-\alpha)(-\alpha t^{-\alpha-1}(d-t)^\alpha-\alpha t^{-\alpha }(d-t)^{\alpha-1})\\
    &\hspace{2.6cm}-\alpha(1-\alpha)t^{-\alpha}(d-t)^{\alpha-1}+\alpha(\alpha-1)t^{1-\alpha}(d-t)^{\alpha-2}+\alpha(\alpha-1) t^{\alpha-2})\\
    &=-\alpha(1-\alpha)d^{-\alpha}\underbrace{(t^{-\alpha-1}(d-t)^\alpha +t^{-\alpha}(d-t)^{\alpha-1}+t^{-\alpha}(d-t)^{\alpha-1}+t^{1-\alpha}(d-t)^{\alpha-2}+t^{\alpha-2})}_{>0}
\end{align*}
on $t\in [2,d-1]$. Thus, $f''(t) < 0$ for $\alpha \in (1/2,1)$ and $f''(t)>0$ for $\alpha >1$. Thus, 
\begin{align*}
    \min_{t\in [2,d-1]}f(t)&=\min\{f(2),f(d-1)\}\quad (\alpha \in (1/2,1))\\
    \max_{t\in [2,d-1]}f(t)&=\max \{f(2),f(d-1)\}\quad (\alpha > 1)
\end{align*}
Now, notice that $f(d-1)=d^{-\alpha}((d-1)^{1-\alpha}+(d-1)^\alpha)$. Thus, it suffices to show that $f(2)\geq 2^{1-\alpha}$ for $\alpha \in (1/2,1)$ and $f(2)\leq 2^{1-\alpha}$ for $\alpha > 1$. 

We begin with the lower bound for $\alpha \in (1/2,1)$. Manipulating,
\begin{align*}
    2^{1-\alpha}&\leq f(2)=d^{-\alpha}(2^{1-\alpha}(d-2)^{\alpha}+2^{\alpha}) \\
    \Leftrightarrow d^\alpha -(d-2)^\alpha &\leq 2^{2\alpha -1}
\end{align*}
Now, consider $h(d)=d^\alpha-(d-2)^\alpha$. Then, 
\begin{align*}
    h'(d)=\alpha(d^{\alpha-1} -(d-2)^{\alpha-1})<0 
\end{align*}
since $t\mapsto t^{\alpha-1}$ is decreasing for $\alpha < 1$. This means that $\max_{d\geq 4}h(d)=h(4)=4^\alpha-2^\alpha $. Thus, it suffices to prove that $4^\alpha-2^\alpha \leq 2^{2\alpha-1}$, or equivalently $g(\alpha)=4^\alpha/2-2^\alpha \leq 0$. Now, for $\alpha \in (1/2,1)$,
\begin{align*}
    g'(\alpha)&=4^\alpha\ln 2 -2^\alpha\ln2=(4^\alpha-2^\alpha)\ln 2 > 0 
\end{align*}
Thus, $\max_{\alpha \in (1/2,1)}g(\alpha)=g(1)=4/2-2=0$, meaning $g(\alpha)\leq 0$ for all $\alpha \in (1/2,1)$, which proves the lower bound case. 

Now, consider the upper bound claim that $f(2)\leq 2^{1-\alpha}$ for $\alpha > 1$. By an analogous argument to before, it suffices to show that $h(d)\geq 2^{2\alpha-1}$. Now, $h'(d)=\alpha (d^{\alpha-1}-(d-2)^{\alpha-1})>0$ since $t\to t^{\alpha-1}$ is increasing for $\alpha > 1$. Then, $\min_{d\geq 4}h(d)=h(4)$. Thus, by a similar manipulation as before, it suffices to show that $g(\alpha)\geq 0$. Now, recall that $g'(\alpha)=(4^\alpha-2^\alpha )\ln 2> 0 $ for $\alpha > 1$ as well. Thus, $\min_{\alpha>1}g(\alpha)=g(1)=0$, proving the upper bound case. 
\end{proof}

\begin{lemma}\label{lemma:m2-left-min}
   Let $d\geq 3$ and $f(t)=\frac{t^{\alpha}+t^{1-\alpha}}{(t+1)^{\alpha}}$. Then, for $\alpha \in (1/2,1)$,
   $$\min_{t\in [1,d-1]}f(t)=\min\{2^{1-\alpha},d^{-\alpha}((d-1)^\alpha+(d-1)^{1-\alpha})\}$$
   while for $\alpha>1$, 
   $$\max_{t\in [1,d-1]}f(t)=\max\{2^{1-\alpha},d^{-\alpha}((d-1)^\alpha+(d-1)^{1-\alpha})\}$$
\end{lemma}
\begin{proof}
Evaluating the boundary, 
\begin{align*}
    f(1)&
    =2^{1-\alpha},
    \quad f(d-1)
    =d^{-\alpha}((d-1)^\alpha+(d-1)^{1-\alpha})
\end{align*}
which are exactly the left and right terms in the claimed equalities. Thus, it suffices to show that $f(t)$ does not have a local minimum on $(1,d-1)$ for $\alpha \in (1/2,1)$ and does not have a local maximum on $(1,d-1)$ for $\alpha >1$.
Differentiating,
\begin{align*}
    f'(t)&=-\alpha (t+1)^{-\alpha-1}(t^\alpha+t^{1-\alpha})+(t+1)^{-\alpha}(\alpha t^{\alpha-1}+(1-\alpha)t^{-\alpha})
    =\underbrace{t^{-\alpha}(t+1)^{-\alpha-1}}_{>0}g(t)
\end{align*}
where $g(t)=\alpha t^{2\alpha-1}-(2\alpha-1)t+(1-\alpha)$. We now consider the sign of $g(t)$, which is also the sign of $f'(t)$. 
\begin{align*}
    g(1)&
    =2(1-\alpha)\\
    g'(t)&=\alpha(2\alpha-1)t^{2\alpha-2}-(2\alpha-1)= (2\alpha-1)(\alpha t^{2\alpha-2}-1)
\end{align*}
For $\alpha \in (1/2,1)$, $g(1)> 0$ and $g'(t)<0$ for $t\geq 1$, since $\alpha t^{2\alpha-2}\leq \alpha <1$. Then, $g$ is positive then negative on $t\geq 1$, as is $f'$, meaning $f$ does not have a local minimum. For $\alpha > 1$, $g(1)<0$ and $g'(t)>0$, since $\alpha t^{2\alpha -2}\geq \alpha > 1$. Then, $g$ is negative then positive on $t\geq 1$, as is $f'$, meaning $f$ does not have a local maximum. Thus, we are done. 
\end{proof}

\begin{lemma}\label{lemma:m1-min}
   For $t\geq 1$ and $d\geq 3$, define the following function:
   $$f(t)=\frac{(d-2+t)^{2\alpha}+(d-2)+t^{2\alpha}}{((d-2+t)^{2}+d-2+t^2)^\alpha}$$
   Then, for $\alpha \in (1/2,1)$,
   $$\min_{t\geq 1}f(t)=\min\{2^{1-\alpha},d^{-\alpha}((d-1)^\alpha+(d-1)^{1-\alpha})\}$$
   while for $\alpha > 1$, 
   $$\max_{t\geq 1}f(t)=\max\{2^{1-\alpha},d^{-\alpha}((d-1)^\alpha+(d-1)^{1-\alpha})\}$$
   where equality for the left and right branch occur only for $t=1$ and $t\to \infty$ respectively, unless $(d,\alpha)=(3,2)$ in which case equality occurs for all $t$. 
\end{lemma}
\begin{proof}
Evaluating the boundary,
\begin{align*}
    f(1)&=\frac{(d-1)^{2\alpha}+d-1}{((d-1)^2+d-1)^\alpha}=\frac{(d-1)((d-1)^{2\alpha-1}+1)}{d^{\alpha}(d-1)^{\alpha}}=d^{-\alpha}((d-1)^{\alpha}+(d-1)^{1-\alpha})\\
    \lim_{t\to \infty}f(t)&=\lim_{t\to \infty}\frac{t^{2\alpha}+t^{2\alpha}}{(t^2 + t^2 )^\alpha}=2^{1-\alpha}
\end{align*}
which are exactly the left and right terms in the claimed equalities. Thus, it suffices to show that $f(t)$ does not have a local minimum on $(1,\infty)$ for $\alpha \in (1/2,1)$ and also does not have a local maximum on $(1,\infty)$ for $\alpha > 1$. As we will see shortly, we are able to do this with the single exception of $(d,\alpha)=(3,2)$, for which $f(t)$ is constant. Taking derivatives,
\begin{align*}
    f'(t)&=\frac{((d-2+t)^{2}+d-2+t^2)^{\alpha}2\alpha((d-2+t)^{2\alpha-1}+t^{2\alpha-1})}{((d-2+t)^{2}+d-2+t^2)^{2\alpha}}\\
    &\quad -\frac{((d-2+t)^{2\alpha}+(d-2)+t^{2\alpha})\alpha ((d-2+t)^{2}+d-2+t^2)^{\alpha-1}(2 (d-2+t)+2t)}{((d-2+t)^{2}+d-2+t^2)^{2\alpha}}\\
    &=2\alpha ((d-2+t)^{2}+d-2+t^2)^{-\alpha-1}g(t)
\end{align*}
where 
\begin{equation}\label{eq:m1-min-g-h}
\begin{aligned}
 g(t)&\equiv ((d-2+t)^{2}+(d-2)+t^2)((d-2+t)^{2\alpha-1}+t^{2\alpha-1})\\&
 \hspace{0.3in} -((d-2+t)^{2\alpha}+(d-2)+t^{2\alpha}) ( (d-2+t)+t)   \\
     &=(d-2)h(t;d-2,2\alpha-1)
\end{aligned}
\end{equation}
where $h(t;n,\beta)=(1-t)(n+t)^{\beta}+(n+t+1)t^{\beta}-(n+2t)$ with $n\geq 1$, $t\geq 1$, and $\beta \in (0,1)$ or $\beta>1$, depending on if $\alpha \in (1/2,1)$ or $\alpha>1$ respectively. The manipulation in \eqref{eq:m1-min-g-h} can be verified by direct expansion. 

Since $\alpha>0$ and $d\geq 3$, we clearly see that ${\rm sign}(f')={\rm sign}(h)$. We then proceed to analyze the sign of $h$ on $t\in [1,\infty)$. We first make the change of variables $x=1/t\in (0,1]$ such that $h(t)=t^{\beta+1}j(1/t)$ where
\begin{align*}
    j(x)\equiv 1+(n+1)x-(1-x)(1+nx)^\beta -(2+nx)x^\beta 
\end{align*}
and $j(0)=j(1)=0$. We desire to show that $j(x)$ has at most one zero on $x\in (0,1)$. Differentiating,
\begin{align*}
    j'(x)&=n+1+(1+nx)^\beta -\beta n(1-x)(1+nx)^{\beta-1} -nx^\beta -\beta(2+nx)x^{\beta-1} \\
    j''(x)&=\beta n(1+nx)^{\beta-1}+\beta n (1+nx)^{\beta-1} -\beta(\beta-1)n^2 (1-x)(1+nx)^{\beta-2} -n\beta x^{\beta-1}-n\beta x^{\beta-1}\\
    &\quad -\beta(\beta-1)(2+nx)x^{\beta-2}\\
    &=\beta x^{\beta-1}k(1/x)
\end{align*}
where
\begin{align*}
    k(t)&\equiv 2n(n+t)^{\beta-1}+(\beta-1)n^2 (1-t)(n+t)^{\beta-2}-(\beta-1)(n+2t)-2n\\
    &=n(n+t)^{\beta-2}A(t;n,\beta)-B(t;n,\beta)
\end{align*}
with $A(t;n,\beta)=(2-n(\beta-1))t+n(\beta+1)$ and $B(t;n,\beta)=2(\beta-1)t+n(\beta+1)$, all of which one can verify by direct expansion.
We first show that $j''(x)$ has at most one zero on $x\in (0,1)$ by showing that $k(t)$ has at most one zero on $t\in (1,\infty)$. Now, consider $k(t)=0$, or equivalently,
\begin{align*}
    1=\ell(t)\equiv \frac{B(t;n,\beta)}{A(t;n,\beta)n(n+t)^{\beta-2}}
\end{align*}
where we divide the two terms in $k(t)$. Then, 
\begin{align*}
    \p_t \ln \ell(t)&=\frac{2(\beta-1)}{2(\beta-1)t+n(\beta+1)} \\&\quad -\frac{n(\beta-2)(n+t)^{\beta-3}((2-n(\beta-1))t+n(\beta+1))+n(n+t)^{\beta-2}(2-n(\beta-1))}{n(n+t)^{\beta-2}((2-n(\beta-1))t+n(\beta+1))}\\
    &=\frac{(\beta-1)m(t)}{(n+t)A(t;n,\beta)B(t;n,\beta)}
\end{align*}
where 
\begin{equation}
    m(t)=2(\beta-2)(n(\beta-1)-2)\cdot t^2 +n(\beta+1)((n-2)(\beta-1)+2)\cdot t+n^2(\beta+1)(n+2-\beta)
\end{equation}
which one can verify by direct expansion.

We now split into cases for $\beta$. Consider $\beta\in (0,1)$. We claim that $k(t)$ has at most one zero on $(1,\infty)$. Since $t>1$, $n\geq 1$, and $\beta \in (0,1)$, 
\begin{align*}
   A(t;n,\beta)=\underbrace{(2-n(\beta-1))}_{\geq 0}t+n(\beta+1)>0 
\end{align*}
Then, if $B(t;n,\beta)\leq 0$, $k(t)>0$, meaning there are no zeros. Thus, instead consider 
$$t\in \mathcal{D}_B=\{t:B(t;n,\beta)>0\}$$
Now, notice that the coefficient of $t^2$ in $m(t)$, $2(\beta-2)(n(\beta-1)-2)$, is positive for $\beta\in (0,1)$. Further, one can check that the discriminant of the quadratic $m(t)$ factors:
\begin{align*}
    \Delta&\equiv (n(\beta+1)((n-2)(\beta-1)+2))^2-4\cdot 2(\beta-2)(n(\beta-1)-2)\cdot n^2(\beta+1)(n+2-\beta)\\
    &=n^2(\beta+1)(\beta-3)(\beta(n+2)-5n-4)(\beta(n+2)-n-4)
\end{align*}
Since $n\geq 1$ and $\beta \in (0,1)$, the first two factors are both positive while the rest are negative, meaning $\Delta < 0$. Thus, $m(t)>0$ everywhere. Then, for $t\in \mathcal{D}_B$, $A,B,m(t)>0$ and $\beta-1<0$, meaning $\ell(t)>0$ and $\p_t\ln \ell(t)<0$. Since $\p_t \ln \ell(t)=\ell'(t)/\ell(t)$, it follows that $\ell'(t)<0$, meaning there is at most $1$ value of $t$ where $\ell(t)=1$, or equivalently $k(t)=0$. 

Now, since $k(t)$ has at most one zero on $(1,\infty)$, $j''(x)$ has at most one zero on $(0,1)$. Now, suppose $j(x)$ had two distinct zeros $r_1,r_2\in (0,1)$ with $r_1<r_2$. Then, since $j(0)=j(1)=0$, by Rolle's theorem, $j'(x)$ must have zeros at some $r_1'\in (0,r_1),r_2'\in (r_1,r_2),r_3'\in (r_2,1)$. Applying Rolle's theorem again, $j''(x)$ must have a zero at some $r''_1\in (r_1',r_2')$ and $r''_2\in (r_2',r_3')$, which contradicts the fact that $j''(x)$ has at most one zero in $(0,1)$. Thus, $j(x)$ has at most one zero on $(0,1)$, meaning $h(t;n,\beta)$ has at most one zero on $(1,\infty)$. Then, since $\beta <1$, 
\begin{align*}
    \lim_{t\to \infty}h(t;n,\beta)=\lim_{t\to\infty}((1-t)(n+t)^\beta+(n+t+1)t^\beta -(n+2t))=\lim_{t\to \infty}(-2t+O(t^\beta))=-\infty
\end{align*}
meaning $h$ is either always negative or positive then negative, as are $g(t)$ and $f'(t)$. Thus, $f(t)$ is either always decreasing or increasing and then decreasing, meaning it does not have a local minimum for $t\in (1,\infty)$. This finishes the $\alpha\in (1/2,1)$, or equivalently $\beta \in (0,1)$, case. 

Now, consider $\beta > 1$, with $n\geq 2$ for now. We first desire to show that $k(t)$ has at most one zero on $(1,\infty)$.

Since $t>1, n\geq 2, $ and $\beta >1$, 
\begin{align*}
   B(t;n,\beta)=2(\beta-1)t+n(\beta+1)>0 
\end{align*}
Suppose $A\leq 0$. Then, $k(t)=n(n+t)^{\beta-2}A-B<0$, meaning $k(t)$ has no zero in this regime. Thus, instead consider
$$t \in \mathcal{D}_A=\{t:A(t;n,\beta)>0\}$$
on which we show that $m(t)>0$, by splitting into cases. 

Suppose $\beta\in (1,1+2/n]\subset (1,2] $ such that $n(\beta-1)\leq 2$. Then, since $n\geq 2$ as well,
\begin{align*}
    m(t)=2\underbrace{(\beta-2)}_{\leq 0}\underbrace{(n(\beta-1)-2)}_{\leq 0}\cdot t^2 +n(\beta+1)\underbrace{((n-2)(\beta-1)+2)}_{>0}\cdot t+n^2(\beta+1)\underbrace{(n+2-\beta)}_{>0}>0
\end{align*}
Now, suppose $n(\beta-1)> 2$, or equivalently $\beta >1+2/n$, and further suppose that $\beta \geq 2$. Then, 
\begin{align*}
    \mathcal{D}_A=\{t:A(t;n,\beta)=(2-n(\beta-1))t+n(\beta+1)>0\}=(1,t_A)
\end{align*}
where $t_A=\frac{n(\beta+1)}{n(\beta-1)-2}$. Then,
\begin{align*}
    m'(t)&=\underbrace{4(\beta-2)(n(\beta-1)-2)}_{\geq 0}+\underbrace{n(\beta+1)((n-2)(\beta-1)+2)}_{>0}>0\\
    m(1)&=2(\beta-2)(n(\beta-1)-2)+n(\beta+1)((n-2)(\beta-1)+2)+n^2(\beta+1)(n+2-\beta)\\
    &=(n+1)(\beta(n^2-4)+n^2+8)\\
    &>0
\end{align*}
where one can check that $m(1)$ factors as such. Then, $m(1)>0$ and $m'(t)>0$ for $t\in(1,t_A)$ imply that $m(t)>0$ for $t\in (1,t_A)$. 

Finally, suppose $n(\beta-1)>2$ and $\beta \in (1,2)$, meaning $n\geq 3$ and $\mathcal{D}_A=(1,t_A)$ as before. Then,
\begin{align*}
    m''(t)&=4(\beta-2)(n(\beta-1)-2)<0\\
    m(1)&=(n+1)(\beta(n^2-4)+n^2+8)>0\\
    m(t_A)&=\frac{2(\beta-2)n^2(\beta+1)^2}{n(\beta-1)-2}+\frac{n^2(\beta+1)^2((n-2)(\beta-1)+2)}{n(\beta-1)-2}+n^2(\beta+1)(n+2-\beta)\\
    &=\frac{n^2(\beta+1)(n+1)}{n(\beta-1)-2}(2(\beta-1)+\underbrace{n(\beta-1)-2}_{>0})\\
    &>0
\end{align*}
which one can verify by direct computation.
Then, since $m''(t)<0$ on $(1,t_A)$, $m(t)\geq \min\{m(1),m(t_A)\}>0$ for any $t \in \mathcal{D}_A=(1,t_A)$. Thus, in all cases, $m(t)>0$. Then, for $t\in \mathcal{D_A}$, $A,B,m(t),(\beta-1)>0$, meaning 
$$\p_t \ln \ell(t)=\frac{(\beta-1)m(t)}{((n+t)AB)}>0,\quad \ell(t)=\frac{B}{An(n+t)^{\beta-2}}>0\Rightarrow \ell'(t)=\ell(t)\p_t\ln \ell(t)>0$$ 
which means that $\ell(t)=1\Leftrightarrow k(t)=0$ has at most one solution for $t\in \mathcal{D}_A$. Thus, $k(t)$ has at most one zero on $(1,\infty)$. By the same argument via Rolle's theorem as before, $h(t)$ has at most one zero on $(1,\infty)$. We next claim that $h(t)$ can not be positive then negative. 

Suppose $n\geq 3$. For $t=1$,
\begin{align*}
    h(1;n,\beta)&=(1-t)(n+t)^\beta +(n+t+1)t^\beta-(n+2t)\bigg|_{t=1}=0\\
    h'(1;n,\beta)&=-(n+t)^\beta +\beta(1-t)(n+t)^{\beta-1}+t^\beta + \beta (n+t+1) t^{\beta-1}-2\bigg|_{t=1}\\
    &=\beta(n+2)-(n+1)^\beta-1\\
    h'(1;n,1)&=0\\
    \p_\beta h'(1;n,1)&=n+2-\ln(n+1)(n+1)^\beta \bigg|_{\beta=1}\\&=n(1-\ln(n+1))+(2-\ln(n+1))\\
    &\leq (1-\ln 4)n+(2-\ln 4)\\
    &<0\\
    \p_\beta^2 h'(1;n,\beta)&=-\ln^2(n+1)(n+1)^\beta \\&< 0 
\end{align*}
Then, since $\p_\beta^2 h'(1;n,\beta)<0$, $\p_\beta h'(1;n,1)<0$, and $h'(1;n,1)=0$, It follows that $h'(1;n,\beta)<0$ for all $\beta > 1$. Then, since $h(1;n,\beta)=0$, there exists $t'$ such that $h(t;n,\beta)<0$ for all $t\in (1,t')$. Thus, $h(t)$ starts negative and thus can not start positive then negative. 

Next, suppose $n=2$. Suppose further that $\beta \geq 2$. Then, 
\begin{align*}
    h'(1;2,\beta)&=\beta(n+2)-(n+1)^\beta -1 \bigg|_{n=2}=4\beta-3^\beta -1<0
\end{align*}
since $3^\beta > 4\beta $ for $\beta \geq 2$. Then, since $h(1;n,\beta)=0$, as argued before, $h(t;2,\beta)$ starts negative and thus can not start positive and then negative. 

Finally, suppose $n=2$ and $\beta \in (1,2)$. Then, 
\begin{align*}
    \lim_{t\to \infty}h(t;2,\beta)&=\lim_{t\to \infty}((1-t)(2+t)^\beta +(t+3)t^\beta -(2+2t))\\
    &=\lim_{t\to \infty}(t^\beta ((1-t)(1+2/t)^\beta +t+3))-O(t))\\
    &=\lim_{t\to \infty}(t^\beta ((1-t)(1+2\beta/t+O(t^{-2})) +t+3))-O(t))\\
    &=\lim_{t\to \infty}(t^\beta (4-2\beta )+O(t))\\
    &=\infty
\end{align*}
meaning $h(t;2,\beta)$ ends positive, and thus can not be positive and then negative. 

Thus, $h(t)$ has at most one zero on $(1,\infty)$ and can not be positive and then negative. Thus, $h(t)$ is either always positive or negative and then positive, as is $g(t),f'(t)$, meaning $f(t)$ does not have a local maximum on $(1,\infty)$. This finishes the $\beta > 1$ case for $n\geq 2$. 

The final case is $\beta > 1$, with $n=1$. Then, one can verify by direct computation that:
\begin{align*}
    k(t)&=2n(n+t)^{\beta -1}+(\beta-1)n^2 (1-t)(n+t)^{\beta-2}-(\beta-1)(n+2t)-2n\bigg|_{n=1}\\
    &=(1+t)^{\beta-2}((3-\beta)t+\beta+1)-(2(\beta-1)t+\beta+1)\\
    k'(t)&=(\beta-2)(1+t)^{\beta-3}((3-\beta)t+\beta+1)+(3-\beta)(1+t)^{\beta-2}-2(\beta-1)\\
    &=(\beta-1)(p(t)-2)
\end{align*}
where $p(t)\equiv (1+t)^{\beta-3}((3-\beta)t+\beta-1)$. Then,
\begin{align*}
    p'(t)&=(\beta-3)(1+t)^{\beta-4}((3-\beta)t+\beta-1)+(3-\beta)(1+t)^{\beta-3}\\
    &=-(\beta-3)(\beta-2)(t-1)(1+t)^{\beta-4}
\end{align*}
which one can also verify by direct computation.
We now show that $h(t)$ is either strictly positive or strictly negative on $(1,\infty)$ via casework on $\beta$. 

Suppose $\beta \in (1,2]$. Then, $p'(t)\leq 0$, since $t> 1$, meaning $k''(t)=(\beta-1)p'(t)\leq 0$. Then,
\begin{align*}
    p(1)&=2^{\beta-3}((3-\beta)+\beta-1)=2^{\beta -2 }<2\Rightarrow k'(1)=(\beta-1)(p(1)-2)<0
\end{align*}
Thus, $k'(t)<0$ for $t\geq 1$. Then,
\begin{align*}
    k(1;\beta )&=2^{\beta-2}((3-\beta)+\beta+1)-(2(\beta-1)+\beta+1)=2^\beta-3\beta+1\\
    k(1;1)&=0\\
    \p_\beta k(1;\beta)&=\ln(2)2^{\beta }-3
\end{align*}
Now, $\p_\beta k(1;1)=0$ has solution $\beta^*\equiv \log_2(3/\ln(2))\approx 2.11>2$, so for $\beta \in (1,2]$, $\p_\beta k(1;1)<0$, 
meaning $k(1;\beta)<0$ for all $\beta \in (1,2]$. Then, since $k'(t)<0$ as well, $k(t;\beta)<0$, meaning $j''(x)<0$. Since $j(0)=j(1)=0$ from before, this implies $j(x)>0$ for $x\in (0,1)$, meaning $h(t)>0$ for $t\in (1,\infty)$, as is $g(t),f'(t)$, meaning $f(t)$ can not have a local maximum on $(1,\infty)$. 

Next, suppose $\beta\in (2,3)$. Then, $p'(t)>0$ for $t>1$, $p(1)=2^{\beta-2}<2$ still, and 
$$\lim_{t\to \infty}p(t)=\lim_{t\to \infty}(1+t)^{\beta-3}((3-\beta)t+\beta-1)=\infty$$
Then, $k'(t)=(\beta-1)(p(t)-2)$ has $k''(t)>0$, $k'(1)<0$, and $\lim_{t\to \infty}k'(t)>0$. This means that $k'(t)$ changes sign exactly once, from negative to positive. Now, consider $k(1;\beta)=2^\beta-3\beta +1$, which we looked at in the previous case. Recall that $\beta^*=\log_2(3/\ln(2))$ is the only extrema. 
\begin{align*}
    k(1;3)&=k(1;1)=0\\
    k(1,\beta^*)&=3/\ln(2)-3\log_2(3/\ln(2))+1\approx -1.01<0
\end{align*}
Now, $\p_\beta k(1;\beta)=\ln(2)2^\beta-3$ must be negative for $\beta < \beta^*$ and positive for $\beta > \beta^*$. Then, since $k(1,\beta^*)<0$, since $k(1;3)=k(1;1)=0$, $k(1;\beta)<0$ for all $\beta \in (1,3)$. Then, since $k(1)<0$ and $k(t)$ decreases then increases, $k(t)$ has at most one zero on $(1,\infty)$. Then, by the Rolle's theorem argument from before, $h(t)$ has at most one zero on $(1,\infty)$. Then,
\begin{align*}
    h(1;1,\beta)&=(1-t)(1+t)^\beta +(t+2)t^\beta -(1+2t)\bigg|_{t=1}=0\\
    h'(1;1,\beta)&=-(1+t)^\beta +\beta(1-t)(1+t)^{\beta-1}+t^\beta +\beta (t+2)t^{\beta-1}-2\bigg|_{t=1}\\
    &=-2^\beta +3\beta-1\\
    &>0
\end{align*}
where we use the fact that we previously showed that $k(1;\beta)=2^\beta -3\beta + 1<0$ for $\beta > 1$. Further,
\begin{align*}
    \lim_{t\to \infty}h(t;1,\beta)&=\lim_{t\to \infty}((1-t)t^\beta(1+\beta/t+O(t^{-2}))+t^{\beta+1}+2t^{\beta}-1-2t)\\
    &=\lim_{t\to \infty}(t^{\beta}(-t+1-\beta +O(t^{-1}))+t^{\beta+1}+2t^{\beta}-1-2t)\\
    &=\lim_{t\to \infty}((3-\beta)t^{\beta}+O(t^{\beta-1}))\\
    &=\infty 
\end{align*}
Since $h(1)=0$, $\lim_{t\to \infty}h(t)=\infty$, $h'(1)>0$, and $h(t)$ has at most one zero on $(1,\infty)$, it follows that $h(t)$ is always positive on $(1,\infty)$. 

Next, suppose $\beta=3$. Then,
\begin{align*}
    h(t;1,3)&=(1-t)(1+t)^3 +(t+2)t^3 -(1+2t)\\
    &=1+3t+3t^2+t^3 -t-3t^2-3t^3 -t^4 +t^4 +2t^3 -1-2t\\
    &=0
\end{align*}
meaning $g(t)=f'(t)=0$, which implies that $f(t)$ can not have a local maximum on $(1,\infty)$. 

Finally, suppose $\beta>3$, and again we analyze the zeros of $k(t)$. Then, 
\begin{align*}
    k(t)=0\Leftrightarrow (1+t)^{\beta-2}((3-\beta)t+\beta+1)=2(\beta-1)t+\beta+1
\end{align*}
Observing that the right hand side and $(1+t)^{\beta-2}$ are both positive, $(3-\beta)t+\beta+1>0$ as well for the root, or equivalently $t\in (1,t')$ where $t'=(\beta+1)/(\beta-3)$. Now, for $\beta > 3$ and $t>1$, $p'(t)<0$, and recall $p(1)=2^{\beta-2}>2$. Further,
\begin{align*}
    p(t')&=\left(1+\frac{\beta+1}{\beta-3}\right)^{\beta-3}\left((3-\beta)\frac{\beta+1}{\beta-3}+\beta-1\right)=-2 \left(1+\frac{\beta+1}{\beta-3}\right)^{\beta-3}<0
\end{align*}
Thus, regarding $k'(t)=(\beta-1)(p(t)-2)$, $k'(1)>0$, $k'(t')<0$, and $k''(t)<0$ for $t\in (1,\infty)$. Thus, $k'(t)$ on $(1,t')$ changes sign exactly once, from positive to negative. Now, consider $k(1;\beta)=2^{\beta}-3\beta +1$. As seen before, $\p_\beta k(1;\beta)>0$ for $\beta > \beta^*\approx 2.11$ and $k(1;3)=0$. Thus, $k(1;\beta)>0$. Then, since $k(1)>0$ and $k(t)$ is increasing then decreasing on $(1,t')$, $k(t)$ has at most one zero in $(1,t')$, and $(1,\infty)$ for that matter, as does $h(t)$ by the Rolle's theorem argument. Then, recall that $h(1;1,\beta)=0$, and
$$h'(1;1,\beta)=-2^{\beta}+3\beta -1=-k(1;\beta)<0$$
From the previous case, we also have that $\lim_{t\to \infty}h(t;1,\beta)=\lim_{t\to \infty}((3-\beta)t^{\beta}+O(t^{\beta-1}))=-\infty$. Thus, if $h(1)=0$, $h'(1)<0$, $\lim_{t\to \infty}h(t)=-\infty$, and $h(t)$ has at most one zero on $(1,\infty)$, it follows that $h(t)<0$ on $(1,\infty)$.

To summarize, we have that for $\beta > 1$ and $n=1$, $h(t)$ is either strictly positive or strictly negative on $(1,\infty)$ as long as $\beta \neq 3$. Then, this is also the case for $g(t)$ and $f'(t)$, meaning $f(t)$ is monotone and thus does not admit local maxima in $(1,\infty)$. The only exception is $\beta =3$, for which $f'(t)=0$. This corresponds to the exceptional equality case for $(d,\alpha)=(3,2)$.
\end{proof}

\end{document}